\documentclass[12pt]{article}

\usepackage[margin=1in]{geometry}
\usepackage{lmodern}

\usepackage{amsmath,amssymb,amsfonts,dsfont}

\usepackage{graphicx}
\usepackage{float}
\usepackage{booktabs}
\usepackage{array}
\usepackage{multirow}
\usepackage{bigstrut}
\usepackage{longtable}
\usepackage{threeparttable}
\usepackage{caption}
\usepackage{subcaption} 
\usepackage{pdflscape}
\usepackage{tikz}
\usetikzlibrary{positioning, arrows.meta}

\usepackage[dvipsnames]{xcolor}
\usepackage{hyperref}
\hypersetup{
  colorlinks=true,
  citecolor=blue,
  urlcolor=blue,
  linkcolor=blue
}

\usepackage[authoryear,longnamesfirst,round]{natbib}
\usepackage{amsthm}
\newtheorem{theorem}{Theorem}
\newtheorem{corollary}{Corollary}

\newtheorem{assumption}{Assumption}
\newtheorem{definition}{Definition}
\newtheorem{proposition}{Proposition}
\newtheorem{remark}{Remark}
\newtheorem{example}{Example}

\usepackage{algorithm}
\usepackage{algpseudocode}

\usepackage{tikz-cd}
\usepackage{afterpage}
\usepackage{comment}
\usepackage{footmisc}
\usepackage{chngcntr}

\title{Better Measurement or Larger Samples?\\ Data Collection for Policy Learning with Unobserved Heterogeneity}
\author{Giacomo Opocher\thanks{\scriptsize University of Bologna, giacomo.opocher2@unibo.it. This paper previously circulated with the title "Policy Learning with Unobserved Heterogeneity". It greatly benefited from the guidance of Silvia Sarpietro and Davide Viviano, and meaningful discussions with Isaiah Andrews, Pietro Biroli, Marc Clos, Toru Kitagawa, Nicola Mastrorocco, Andrea Mattozzi, Kirill Ponomarev, Christoph Rothe, Rahul Singh, and Giulio Zanella. I also thank all seminar participants at the University of Bologna, ETH Zurich, the CEPR Job Market Bootcamp, Brown University, Harvard University, University of Mannheim, University of Bonn, and OCIS for their insightful comments. All mistakes are my own.}}
\date{\today}

\begin{document}
{\maketitle
\vspace{-4.2em}
\thispagestyle{empty}
\begin{abstract} \textbf{Abstract.}
    Empirical research shows that individuals' responses to treatments vary along latent characteristics, such as innate ability or motivation. 
    Therefore, a policymaker seeking to maximize welfare may consider designing policies based on observed characteristics and estimated latent traits. 
    I characterize how the estimates' precision affects the worst-case performance of policies deriving rate-sharp regret bounds for assignment rules that include or exclude them, highlighting new trade-offs with the policy space complexity. 
    I then study how a policymaker can solve such trade-offs by designing tailored data collections and derive a sufficient condition for a collection plan to be minimax optimal. 
    In an empirical application in development economics, I show that including a proxy for entrepreneurs' business skills in targeting cash transfers increases welfare by $5\%$, and halves the probability of generating welfare losses. 
    Moreover, I estimate the optimal allocation of resources between improving the precision of the proxy via repeated measurements, and increasing sample size. 
\end{abstract}}

\section{Introduction}
Governments and institutions increasingly rely on individualized treatment rules to allocate interventions in heterogeneous populations. From targeting cash transfers to assigning job training, the goal is to identify subgroups that benefit most from a given policy, based on observable characteristics. 
Recent advances in policy learning formalize this task as the problem of estimating assignment rules that maximize expected welfare, using experimental or observational data \citep[e.g.][]{kitagawa_who_2018,athey_policy_2021}.

A large body of empirical and theoretical research highlights that individuals' responses to treatments may depend not only on covariates such as age or income, but also on latent characteristics such as motivation, prior experience, or ability.\footnote{e.g. \cite{heckman_policy-relevant_2001, heckman_structural_2005}.} 
In structural econometric settings, these unobservables are often modeled through fixed effects or individual-specific components, which can be estimated under repeated observations or panel structures.\footnote{e.g. \cite{Wooldridge_2005, shoser_2020}.} 
Alternatively, applied researchers measure proxies of the unobserved factors and consider treatment effect variation along their values. 
For example, performance indicators have been used as proxies for workers' skill level to assess the impact of new technologies on workers' productivity;\footnote{e.g. \cite{qje_AI}.} community ratings and psychometric measures of business skills have been used to target resources to high-growth microentrepreneurs in developing countries.\footnote{e.g. \cite{hussam_2022, bryan_2024}.} 

As a result, a policymaker interested in maximizing social welfare may decide to assign policies based on the estimated values of these relevant latent traits. This decision problem raises two questions. 

First, under what conditions is leveraging such a source of information to assign treatments welfare-improving? \\ 
To shed light on this question, I show that the proxy's measurement error propagates into the decision problem. 
Therefore, for its inclusion to improve worst-case performance, the variation in treatment effects explained by the underlying latent factor must outweigh (i) the additional estimation error introduced and (ii) the increase in policy space complexity. 
To study this trade-off formally, I derive rate-sharp regret bounds for rules that ignore unobserved heterogeneity (Covariate-Based rules) and rules that acknowledge its presence by including the estimate, or proxy ($\hat{a}$-Augmented rules). 
This comparison is delivered by a simple theoretical innovation. I define regret as the expected welfare loss of any estimated rule relative to an oracle that observes the true latent factor. 
This provides a common benchmark across policy classes and makes the comparison between the two classes of interest meaningful.

Because the proxy's estimation error affects the policy's worst-case performance, the policymaker may consider to invest in its precision, for instance refining measurement, designing incentive-compatible elicitation mechanisms, or collecting richer datasets to train predictive models.\footnote{Examples include (i.) acquiring satellite images at higher resolution \citep{satellite_data}, or repeating measurement \citep{hussam_2022}; (ii.) designing a Becker-DeGroot-Marschak meachanism \citep{BDM_method}; (iii.) collecting data along the long dimension of a panel dataset when estimating $A_i$ with a fixed or random effects model.}
However, under a finite budget, such an investment implies a smaller sample size to learn the optimal policy, leading to a higher welfare loss due to the increase in policy space complexity. 

This tension raises the second question. How much should the policymaker invest in the proxy's precision relative to sample size to maximize the policy's performance? \\
I study the design of data collections for policy learning when the policymaker faces a fixed budget. 
I show that when latent heterogeneity in treatment effects and returns to investment in the proxy's precision are sufficiently high, it is optimal to devote resources to the measurement (or estimation) of the latent factor. By contrast, when it is too costly to improve on the proxy's precision, or its relevance is limited, it is optimal to allocate the budget to enlarging the policy-learning sample and to rely on treatment rules based only on standard covariates. 
I leverage the regret bounds to derive the threshold conditions that separate these cases yielding a sufficient condition for minimax optimal budget allocation.

In line with the econometric literature in policy learning, I adopt the minimax approach to provide theoretical guarantees on regret \citep[see e.g.][]{manski_statistical_2004}, and derive optimal data collection plans \citep[see e.g.][]{epanomeritakis2025, breza2025}. 
To provide practical guidance for applied researchers that do not adopt a minimax perspective, I also propose two sample-splitting procedures that can be implemented in a given empirical setting to provide evidence on: (i) the ranking of treatment rules that ignore or incorporate unobserved heterogeneity; and (ii) how to scale up data collections optimally by allocating resources between measuring (or estimating) the proxy and increasing the sample used to learn the policy.

I apply these new procedures to the context studied in \cite{hussam_2022}. 
The authors conduct a cash transfer randomized controlled trial in rural India and present a new proxy for micro-entrepreneurs' business skills based on the rankings entrepreneurs give to each other. 
They call this measure community rankings. 
The main result they report is that community rankings improve targeting of cash transfers. 
First, I confirm the original result by showing that it indeed increases average welfare by $3\%$, and reduces the probability of producing welfare losses by a third compared to scaling up the intervention using only covariates. 
The proxy was based on the average assessment of five separate rankers. 
This feature allows me to report two other key findings. First, I ignore the collection cost and show that the welfare gain would have been substantially smaller, had the number of rankers been lower while keeping fixed sample size. 
Second, I pretend the data from the study were used as a pilot to guide bigger data collections and I estimate the optimal allocation of finite budgets between the number of rankers and sample size of the RCT. I show that for limited budgets it is optimal to select two rankers instead of five in favor of sample size.

The rest of the paper is organized as follows. In section \ref{sec:literature}, I review the related literature and describe the main contribution of this paper; in section \ref{sec:setting}, I introduce the formal setting, definitions, and main assumptions; in section \ref{sec:regret_bounds}, I derive the regret bounds for Covariate-Based, and $\hat{a}$-Augmented rules; in section \ref{sec:data_collection}, I study the data collection problem; in section \ref{sec:emp_application}, I present the empirical application; section \ref{sec:conclusion} concludes.

\subsection{Contribution to the Literature}\label{sec:literature}
This paper contributes to the literature on policy learning connecting new regret bounds for policy rules that include or ignore unobserved heterogeneity in treatment effects to the design of minimax optimal data collection plans. 
This connection is made possible by a new definition of regret that fixes as a benchmark for all classes an oracle that directly observes the latent factor and has complete knowledge of the causal structure underlying the data. 
This simple theoretical innovation allows one to derive non-trivial rate-sharp regret bounds for both classes and to reduce the data collection problem to a tractable budget allocation problem between competing objectives.
These theoretical results come with practical, data-driven procedures to rank policy rules that ignore or incorporate unobserved heterogeneity and estimate the optimal allocation of budget between measuring (or estimating) latent factors more accurately and increasing sample size. To the best of my knowledge, this is the first paper that combines (i) the policy learning problem when policy-relevant variables are estimated or observed with error, with (ii) the resulting trade-offs involved in designing data collections. 

The problem of learning optimal treatment assignment rules has attracted attention in economics, statistics, and machine learning. 
Foundational work by \cite{manski_statistical_2004} framed the problem of treatment choice as an empirical risk minimization problem, considering regret as a key evaluation metric. \cite{kitagawa_who_2018} formalized empirical welfare maximization as a framework for optimizing treatment rules with controlled complexity, deriving minimax regret bounds for policy classes with finite complexity. 
\cite{athey_policy_2021} extended this framework by focusing on observational studies. 
More recent contributions explore extensions beyond the standard approach: \cite{viviano_fair_2024} and \cite{kitagawa_equality-minded_2021} formalize notions of fairness and equality in policy learning; \cite{viviano_policy_2024} studies treatment assignment under network interference; \cite{kitagawa2025} studies the case in which the set of covariates that is relevant in explaining treatment effects heterogeneity is wider than the set used for targeting. 
One closely related paper is \cite{mbakop_model_2021}. It proposes the Penalized Welfare Maximization (PWM) framework, which addresses model selection in treatment choice by penalizing policy complexity. 
The main similarity relates to the formulation of the problem: both papers consider the problem of optimally selecting the set of policy-relevant variables. 
However, PWM's guarantees would not apply trivially to the context of unobserved heterogeneity, as it does not explicitly consider noise propagation in the decision problem, which is the main focus of the present work. Moreover, they do not frame the data collection problem or study the trade-offs involved in it. 

The econometric literature has long recognized that treatment effect heterogeneity often arises from unobserved factors. 
Seminal work by \citet{heckman_policy-relevant_2001,heckman_structural_2005} introduced the concept of essential heterogeneity and the marginal treatment effect (MTE), showing how unobserved traits influence both treatment selection and gains. 
This framework highlights that ignoring latent heterogeneity can bias causal inference and limit the effectiveness of policy rules. 
Building on these insights, a large body of work has focused on the identification and estimation of treatment effects under limited exogeneity.\footnote{For instance, \cite{abadie_instrumental_2002}, and \cite{chernozhukov_iv_2005} develop IV-based methods for estimating heterogeneous effects, while \cite{frolich_unconditional_2013}, and \cite{dhaultfoeuille_identification_2015} extend these approaches to continuous treatments and nonparametric settings.} 
Recent work has begun to explore policy learning under unobserved confounding. \cite{kallus_confounding-robust_2018} proposes minimax regret bounds that hedge against hidden bias, while \cite{cui_semiparametric_2021} adapts instrumental variables methods to estimate optimal treatment rules. 
Proximal causal inference approaches \citep[see][for a review]{tchetgen_introduction_2024} use proxies to adjust for unobserved confounders. 
This paper takes a different perspective. 
I show that even when standard identification issues from unobserved heterogeneity, such as differential compliance, selection into treatment assignment, or spillovers, are not present, an important theoretical trade-off emerges from the fact that relevant unobserved traits need to be estimated or measured with error. Finally, none of these papers studies the data collection problem.

This paper is also related to the long standing econometric literature on measurement error \citep[e.g.][]{BOUND20013705, decon_repeated_2008, BATTISTIN2014707, deaner_2023, rahul_corrupted_data}. 
Foundational surveys such as \citet{nonlinear_review_2011} and \citet{review_mesurement_error} review how mismeasured covariates affect identification and estimation in linear and nonlinear models, and emphasize the role of auxiliary information such as validation data, repeated measurements, and instrumental variables. 
In particular, the model considered in the present paper can be classified as \textit{weakly}-classical measurement error model \citep[see Page 343,][]{review_mesurement_error} as I allow for the proxy to be biased and for the distribution of measurement error to vary with covariates, but I do not allow for it to vary with the true value of the latent factor. 
The present paper differs from this literature in both object and question. Rather than studying the bias induced by measurement error in the estimation of treatment effects, I take identification of welfare-relevant treatment effects as given and analyze how proxy noise propagates into the policy-learning decision problem. 
This shifts attention from identification and bias correction to regret bounds, minimax dominance across policy classes, and the resulting trade-off between improving proxy precision and increasing the sample size used to learn the policy.

Finally, this paper contributes to an emerging econometric literature on data collection problems and experimental design \citep[e.g.][]{manski_2017, gechter2024selectingexperimentalsitesexternal, epanomeritakis2025,breza2025} by formalizing the problem of designing data collection plans tailored to the problem of learning optimal policies when unobserved heterogeneity is policy-relevant.

\section{Formal Setting, Definitions, and Main Assumptions}\label{sec:setting}

\paragraph{Data Generating Process.} 
Consider the random vector $(X_i,A_i)$, with $(x,a) \in \mathcal{X}\times\mathcal{A}$, $\mathcal{X} \subseteq \mathbb{R}^d$, and $\mathcal{A} \subseteq \mathbb{R}$. 

Define $\mathcal{D} = \{0,1\}$ a binary treatment and $D_{i} \in \mathcal{D}$ the treatment indicator. Consider the outcome $Y_i$ and denote with $(Y_i(0), Y_i(1)) \sim P_Y$ the potential outcomes in case $D_{i} = 0$ or $1$ respectively. Define the treatment effect $\tau_i:=Y_i(1)- Y_i(0)$.  Denote with $Y_i = (1-D_i)\cdot Y_i(0) + D_i\cdot Y_i(1)$ the observed potential outcome. We observe one realization of $(Y_i, X_i,D_i)\sim_{\text{i.i.d.}}P_{Y,X,D}\in \mathcal{P}$ for all $i \in S_n$ where $S_n$ is a random sample of $n$ units. 

We also observe a proxy, or estimate of $A_i$, $\hat{A}_i$ that takes values $\hat{a} \in \hat{\mathcal{A}}$. This can be a direct measurement with error, or a data-dependent estimate.

\paragraph{Policy Rules and Policy Classes.}
A policy rule is a function that maps a general set of characteristics $Z_i$ into the target set: $G: \mathcal{Z} \rightarrow \{0,1\}$.

I define as Covariate-Based (CB) the rules that consider only the values of observed covariates to identify targets: \( G(x): \mathcal{X}\rightarrow \{0,1\} \), as $a$-Augmented (${a}$-CB for later reference) rules, the rules that also include unobserved variables: \( G(x,a) : \mathcal{X} \times \mathcal{A} \rightarrow \{0,1\} \), and as feasible $a$-Augmented rules ($\hat{a}$-CB for later reference), rules that leverage observed covariates and estimates of unobserved variables: \( G(x,\hat{a}) : \mathcal{X} \times \hat{\mathcal{A}} \rightarrow \{0,1\} \).  Let \( \mathcal{G}_x \), \( \mathcal{G}_{x,a} \), and \( \mathcal{G}_{x,\hat{a}} \) denote the respective policy classes defined as collections of rules. I indicate with $\mathcal{G}_z = \{G(z)\}$ the class of policy rules that belong to any of the three types described above. I denote with $v_z = \text{VC}(\mathcal{G}_z)$ the VC-dimension of the class $\mathcal{G}_z$. 

We restrict our attention to the classes of parametric policies defined as:
\begin{equation}
    \mathcal{G}_{z}^\theta:=\{G_\theta(z):= \mathbf{1}\{s_\theta(z)\geq 0\}\}
\end{equation}
where $\theta \in \Theta_z$ and $s: \mathcal{Z} \rightarrow \mathbb{R}$. 

\begin{remark}
    All classes considered in Examples 2.1, 2.2, 2.3 \citep{kitagawa_who_2018}, Examples 2.2 and 2.3 \citep{mbakop_model_2021}, and the examples provided in section 2.2 \citep{athey_policy_2021} fit inside this class.
\end{remark}

Moreover, define the conditional average treatment effect function:
\begin{equation}
    \tau(z): \mathcal{Z} \rightarrow \mathcal{Y}\quad \text{such that}\quad \tau(z) = \mathbb{E}_P[\tau_i | Z_i=z]
\end{equation}
And the first best rule:
\begin{equation}
    {G}^{FB}(z) := \mathbf{1}\{\tau(z) \geq 0\} 
\end{equation}

\paragraph{Welfare.}
Population welfare is defined as:
\begin{equation}
    W(G_\theta(Z_i)) := \mathbb{E}_P\left[ Y_i(1) \cdot G_\theta(Z_i) + Y_i(0) \cdot (1-G_\theta(Z_i)) \right]
\end{equation}

The best-in-class rule is defined as the rule that directly maximizes population welfare. Formally,
\begin{equation} \label{eq:oracle_rule}
    G_{\theta}^*(Z_i) := \arg\max_{G(Z_i) \in \mathcal{G}^\theta_z} W(G(Z_i))
\end{equation}
We cannot solve this problem directly because we observe only a random sample of the population of interest and we lack knowledge of the causal law underlying $(Y_i(0),Y_i(1))$. Therefore, following \cite{kitagawa_who_2018}, we rely on its empirical analog and estimate the empirical optimal rule:
\begin{equation} \label{eq:ewm}
    \hat{G}_\theta(Z_i) := \arg\max_{G(Z_i) \in \mathcal{G}^\theta_z} \left\{ W_n(G(Z_i)) := \frac{1}{n} \sum_{i=1}^n \left[ \frac{Y_i D_i}{e(Z_i)} \cdot G(Z_i) + \frac{Y_i (1 - D_i)}{1 - e(Z_i)} \cdot (1-G(Z_i)) \right] \right\}
\end{equation}
where $e(Z_i)$ is the propensity score given $Z_i$. 

We evaluate the performance of estimated treatment rules in comparison with an oracle that observes both the values of $X_i$ and $A_i$:
\begin{equation}\label{eq:regret}
    R(\hat{G}_\theta(Z_i)) := \mathbb{E}_{P^n}[W(G_{\theta}^*(X_i,A_i)) - W(\hat{G}_\theta(Z_i))]
\end{equation}

\begin{remark}
    \cite{kitagawa_who_2018} and subsequent literature define regret within class:
    \begin{equation}
        \mathbb{E}_{P^n}[W(G_{\theta}^*(Z_i)) - W(\hat{G}_\theta(Z_i))]
    \end{equation}
    The new definition of regret in \eqref{eq:regret} is necessary to compare different classes to the same benchmark. Moreover, it is the natural benchmark when deriving optimal data collection plans: with infinite collection effort we could (i) directly observe $A_i$ by investing an infinite amount on the measurement (or estimation) of $A_i$ and (ii) directly compute $G_\theta^*(X_i,A_i)$ by investing in a sample $S_n$ of infinite size. By contrast, with finite collection effort we need to choose how to allocate budget between these two competing objectives.
\end{remark}

\subsection{Main Assumptions} 
The main assumptions can be divided into assumptions on the data generating process (Assumption \ref{ass:kt18}), on the generating process of $\hat{A}_i$ (Assumption \ref{ass:proxy}), and on the policy space (Assumption \ref{ass:score_funct}).
\begin{assumption}[Data generating process] \label{ass:kt18} $ $
\begin{enumerate}
    \item \textbf{Bounded Outcomes} - There exists $M<\infty$ such that the support of the outcome variable $\mathcal{Y} \subseteq [-M/2, M/2]$. 
    \item \textbf{Stratified Random Assignment} - Treatment assignment is such that $(Y_i(0), Y_i(1), \hat{A}_i) \perp D_i |X_i$. Propensity scores $e(X_i)$ are known.
    \item \textbf{Strict Overlap} - There exists $k \in (0,1/2)$ such that $e(x)\in[k,1-k]$ for all $x \in \mathcal{X}$. 
\end{enumerate}
\end{assumption}

Assumption \ref{ass:kt18}.i implies that both potential outcomes, and thus treatment effects, are uniformly bounded in absolute value by $M$. Boundedness is a standard condition in the statistical learning literature as it enables the use of uniform concentration inequalities \citep[see, e.g.][]{hoeffding_probability_1963,van_der_vaart_weak_2023}. Assumption \ref{ass:kt18}.ii characterizes a quasi-experimental environment in which treatment assignment is independent of potential outcomes and $\hat{A}_i$ conditional on observed covariates. Moreover, the potential outcome of each unit $i$ depends only on their own treatment status, and propensity scores are known. Finally, Assumption \ref{ass:kt18}.iii is standard in the causal inference literature and guarantees that all units have a positive probability of receiving either treatment or control. 
\begin{example}
    Assumptions \ref{ass:kt18}.2 and \ref{ass:kt18}.3 are satisfied by design in stratified randomized controlled trials \citep[see e.g.][for a definition]{gerber_green}. 
\end{example}

\begin{assumption}[Measurement error-based $\hat{A}_i$]\label{ass:proxy} $ $
\begin{enumerate}
    \item \textbf{Proxy Representation} - Let $\hat{A}_i$ be written as $\hat{A}_i = A_i + \varepsilon_i$.
    \item \textbf{Noise Distribution} - $\varepsilon_i|X_i \sim F_{\varepsilon|X_i}$. Moreover, $\varepsilon_i \perp A_i|X_i$.
\end{enumerate}
\end{assumption}

Assumption \ref{ass:proxy}.1 imposes that $\hat{A}_i$ is produced by a measurement with error. In particular, it imposes additive separability between noise and signal. Assumption \ref{ass:proxy}.2 imposes that the measurement error is random conditional on covariates. 
As a whole, Assumption \ref{ass:proxy} allows $\hat{A}_i$ to be biased and its error's distribution to vary across covariate values, while requiring the measurement error to be independent of the true values, conditional on the covariates. 
In Appendix \ref{app:external_proxy} I extend Assumption \ref{ass:proxy} for the case where $\hat{A}_i$ is estimated from external data, rather than measured with error.

\begin{example}
    Night-time light intensity from remote sensing is frequently used as a proxy for local economic activity \citep[see e.g.][]{satellite_data, satellite_appl}. One source of measurement error allowed by Assumption \ref{ass:proxy} is adversarial atmospheric conditions. 
    Assumption \ref{ass:proxy} allows this source of error to be correlated with local characteristics (e.g. geography). It is not allowed to vary with the true economic activity within local characteristics.
\end{example}

\begin{example}
    Survey questions are frequently used as proxies for economic and psychological latent traits such as business skill or cognitive ability \citep[see e.g.][]{stancheva2023, hussam_2022}. One common source of measurement error is the experimenter demand effect, i.e. the framing of the survey question may induce the subject to over- or under-state a given trait of interest. 
    Assumption \ref{ass:proxy} allows this source of error to vary along subjects' observed characteristics.
    It is not allowed to vary with the true value of the underlying trait, or to be correlated with the answers of other subjects.
\end{example}

\begin{assumption}[Policy class restrictions]\label{ass:score_funct} 
$ $
\begin{enumerate}
    \item \textbf{VC Class} - The policy class $\mathcal{G}^\theta_{z}$ has finite VC-dimension $v^\theta_z<\infty$.
    \item \textbf{Flexibility} - $\exists\ \tilde{\theta} \in \Theta_x$ such that $\operatorname{sign}(s_{\tilde{\theta}}(X_i)) = \operatorname{sign}(\tau(X_i))$ for all $x \in \mathcal{X}$.
    \item \textbf{Margin Condition} - There exists a constant $\kappa>0$ such that, for all $t\geq0$:
    \begin{equation}
        \sup_{\theta \in \Theta}\mathbb{P}(|s_\theta(X_i,A_i)|<t|X_i=x) \leq \kappa t \quad \forall\ x \in \mathcal{X}
    \end{equation} 
    \item \textbf{Lipschitz Continuity} - There exists a constant $L_s$ such that:
    \begin{equation}
        \sup_{\theta\in\Theta,\ (x,a) \in \mathcal{X}\times \mathcal{A}}|s_{\theta}(x,a)- s_{\theta}(x,a+\gamma)| \leq L_s |\gamma|
    \end{equation}
\end{enumerate}
\end{assumption}

Assumption \ref{ass:score_funct}.1 restricts the complexity of the policy class by ensuring that it cannot shatter arbitrarily large sets. The use of VC-dimension as a complexity measure in policy learning was introduced in \cite{kitagawa_who_2018}, and has been widely adopted by the subsequent literature. 
Assumption \ref{ass:score_funct}.2 requires the policy class to be \textit{flexible} enough to contain the true CATE function, or the CATE function to be \textit{simple} enough to be contained inside the policy class. 
This assumption is the most restrictive in the set considered. Note that it is only needed to simplify the regret bound for covariate based rules which otherwise would carry an additional term that cannot be bounded non-trivially. 
I defer a more detailed discussion to the results section and Appendix \ref{app:formal_proofs}. 
Assumption \ref{ass:score_funct}.3 rules out degenerate distributions that place all the probability mass close to the region where the score function $s_\theta(X_i,{A}_i)$ is equal to zero. Assumption \ref{ass:score_funct}.4 rules out score functions that are not Lipschitz continuous.

\begin{example}
    If $(X_i, A_i)$ follow a joint normal distribution and $G_\theta(x,\hat{a})$ is defined as a generalized threshold rule \citep[see e.g. Example 2.2][]{kitagawa_who_2018}, then Assumptions \ref{ass:score_funct}.1, \ref{ass:score_funct}.3 and \ref{ass:score_funct}.4 are satisfied.
\end{example}

\begin{example}
    If $(X_i, A_i)$ follow a joint uniform distribution and $G_\theta(x,\hat{a})$ is defined as a rectangular rule \citep[see e.g. Example 2.3][]{mbakop_model_2021}, then Assumptions \ref{ass:score_funct}.1, \ref{ass:score_funct}.3 and \ref{ass:score_funct}.4 are satisfied.
\end{example}

\section{Learning Policies with Unobserved Heterogeneity}\label{sec:regret_bounds}
In this section, I present the regret bounds for Covariate-Based (section \ref{sec:rb_cb}), and $\hat{a}$-Augmented (section \ref{sec:rb_aug}) rules. In section \ref{sec:minimax_comp} I illustrate the minimax comparison.

\subsection{Performance when Ignoring Unobserved Heterogeneity}\label{sec:rb_cb}
\begin{theorem}[Regret Bound for Covariate-Based Rules]\label{thm:cb_upper}
Fix $\sigma_0>0$. Let $\mathcal{P}(\sigma_0)$ denote the class of data-generating processes $P$ satisfying Assumption \ref{ass:kt18}, and such that 
\begin{equation} 
\bar{\sigma}_{\tau|x}(P) := \mathbb{E}_{X}\left[ \sqrt{\mathbb{V}_A(\tau(X_i,A_i)\mid X_i)} \right] \le \sigma_0 
\end{equation} 
The regret of any CB policy class \( \mathcal{G}^\theta_x \) that satisfies Assumption \ref{ass:score_funct}.1, satisfies:
\begin{equation}\label{eq:cb_upper}
\sup_{P\in \mathcal{P}(\sigma_0)} \mathbb{E}_{P^n}[W(G^*_{\theta}(X_i,A_i)) - W(\hat{G}_\theta(X_i))] 
\leq C_1\frac{M}{k} \sqrt{\frac{v_x^\theta}{n}} 
+ \sigma_0 + \Delta(s,\Theta_x)
\end{equation}
where $C_1>0$ is a universal constant, and
\begin{equation}
    0 \leq \Delta(s,\Theta_x)\leq M\cdot \mathbb{P}_P(\mathbf{1}\{s_{\theta^*}(X_i)\geq 0\}\neq \mathbf{1}\{\tau(X_i)\geq0\})
\end{equation}
Moreover, if $\mathcal{G}_x^\theta$ also satisfies Assumption \ref{ass:score_funct}.2, then $\Delta(s,\Theta_x)=0$.
\end{theorem}
The formal proof is reported in Appendix \ref{app:formal_proofs}. Theorem \ref{thm:cb_upper} introduces a bound on the regret for CB rules arising from (i) completely ignoring the source of unobserved heterogeneity, (ii) the lack of complete knowledge on the counterfactual outcomes. The bound in Eq. \ref{eq:cb_upper} decomposes regret into a statistical error term diminishing with sample size that equals the bound in Theorem 2.1 \citep{kitagawa_who_2018}, and an approximation error term due to (i) ignoring unobserved heterogeneity ($\bar{\sigma}_{\tau|x}$, upper bounded by $\sigma_0$) and (ii) considering an assignment rule that is less flexible compared to the CATE ($\Delta(s,\Theta_x)$). Note that, under Assumption \ref{ass:score_funct}.2, this third term equals zero. 

\begin{theorem}[Minimax lower bound for Covariate-Based rules] \label{thm:cb_lower} Let $\mathcal{P}(\sigma_0)$ be defined as in Theorem \ref{thm:cb_upper}.
Under Assumption \ref{ass:kt18}, and for any class $\mathcal{G}_x^\theta$ that satisfies Assumption \ref{ass:score_funct}, 
\begin{equation}\label{eq:cb_minimax_lower} 
\inf_{\{\hat{G}_\theta(X_i)\}} \sup_{P \in \mathcal{P}(\sigma_0)} R(\hat{G}_\theta(X_i)) \geq C_2\frac{M}{k}\sqrt{\frac{v_x^\theta}{n}} + C_3\sigma_0
\end{equation} 
where $C_2>0$ and $C_3>0$ are universal constants. 
\end{theorem}
The formal proof is reported in Appendix \ref{app:formal_proofs}. Theorem \ref{thm:cb_lower} establishes that the regret of Covariate-Based policy rules is bounded below by the sum of a statistical term of order $\sqrt{v^\theta_x/n}$ and an approximation term proportional to the residual variation in treatment effects unexplained by observed covariates, ${\sigma}_0$. Combined with the upper bound in Theorem \ref{thm:cb_upper}, this result implies that the regret bound for Covariate-Based rules is minimax sharp up to constants over the class $\mathcal P(\sigma_0)$. 

\subsection{Performance when Including Noisy Measures} \label{sec:rb_aug}
Define the root Mean Squared Errror (rMSE) of $\hat{A}_i$ as:
\begin{equation}
        \operatorname{rMSE}(\hat{A}_i) := \sqrt{{\mathbb{E}_P\left[ (\hat{A}_i-A_i)^2 \right]}}
\end{equation}

\begin{theorem}[Regret Bound for $\hat{a}$-Augmented Rules]\label{thm:ha_upper}
    Fix $\rho>0$. Let $\mathcal{P}(\rho)$ denote the class of data-generating processes $P$ satisfying Assumptions \ref{ass:kt18} and \ref{ass:proxy}, and such that $\operatorname{rMSE}(\hat{A}_i)\le \rho \leq 1/(2\kappa)$.
    The regret of any $\hat{a}$-CB policy class \( \mathcal{G}^\theta_{x,\hat{a}} \) that satisfies Assumption \ref{ass:score_funct}.1, \ref{ass:score_funct}.3, and \ref{ass:score_funct}.4, satisfies:
    \begin{equation}
        \sup_{P \in \mathcal{P}(\rho)} \mathbb{E}_{P^n}[W(G^*_{\theta}(X_i,A_i))- W(\hat{G}_{\theta}(X_i,\hat{A}_i))] \leq C_1 \frac{M}{k}\sqrt{\frac{v_{x,\hat{a}}^\theta}{n}} +  M \kappa L_s \rho
    \end{equation}
    where $C_1$ is a universal constant.
\end{theorem}
The formal proof is reported in Appendix \ref{app:formal_proofs}. 
Theorem \ref{thm:ha_upper} introduces a bound on the regret for $\hat{a}$-CB rules arising from (i) not observing the unobserved factor $A_i$, and (ii) the lack of complete knowledge on the counterfactual outcomes. 
This bound is composed of the bound proposed by \cite{kitagawa_who_2018} plus a constant that depends on the class of rules $\mathcal{G}^\theta_{x,\hat{a}}$ through the Lipschitz and margin constants (see Assumption \ref{ass:score_funct}) times the root MSE of $\hat{A}_i$. 
The proof is composed of the following steps. 
First, regret can be decomposed into the sum of the distance between an oracle that observes $A_i$ (full information) and an oracle that observes $\hat{A}_i$ (partial information), and the distance between the latter and the feasible rule. Because of Assumptions \ref{ass:kt18}, \ref{ass:proxy}, and \ref{ass:score_funct}.1, this second term can be bounded by the bound in Theorem 2.2 \citep{kitagawa_who_2018}. Because of Assumption \ref{ass:score_funct}.3, the first term can be bounded by the probability of disagreement between the two oracles scaled by $M$. Because of Assumption \ref{ass:score_funct}.4 such probability can be bounded by a multiple of the expected absolute difference between $\hat{A}_i$ and $A_i$, which in turn can be bounded by the (upper bound of the) rMSE of $\hat{A}_i$.

\begin{theorem}[Minimax lower bound for $\hat{a}$-Augmented rules]
\label{thm:ha_lower}
Let $\mathcal{P}(\rho)$ be defined as in Theorem \ref{thm:ha_upper}.
Then, for any class $\mathcal G^\theta_{x,\hat{a}}$ that satisfies Assumption \ref{ass:score_funct}, 
\begin{equation}\label{eq:ha_lower}
\inf_{\{\hat{G}_{\theta}(X_i,\hat{A}_i)\}}
\sup_{P\in \mathcal P(\rho)}
\mathbb E_{P^n}\!\left[W(G^*_{\theta}(X_i,A_i))-W(\hat{G}_{\theta}(X_i,\hat{A}_i))\right]
\ge
C_4 \frac{M}{k}\sqrt{\frac{v^\theta_{x,\hat a}}{n}}
+
C_5  M \kappa \rho
\end{equation}
where $C_4>0$ and $C_5>0$ are universal constants.
\end{theorem}
The formal proof is reported in Appendix \ref{app:formal_proofs}. Theorem \ref{thm:ha_lower} shows that the regret of $\hat{a}$-Augmented policy rules is bounded below by the sum of a statistical term of order $\sqrt{v^\theta_{x,\hat{a}}/n}$ and an irreducible estimation-error term proportional to the estimation error in the proxy, $\rho$. 
Combined with the upper bound in Theorem \ref{thm:ha_upper}, this result establishes that the regret bound for $\hat{a}$-Augmented rules is minimax sharp up to constants over the class $\mathcal{P}(\rho)$. 
In particular, even with infinite data, imperfect observation of the latent factor $A_i$ induces a non-vanishing welfare loss whenever $\rho>0$, reflecting a fundamental limit to the gains from incorporating noisy estimates of unobserved heterogeneity into policy learning.

In Appendix \ref{app:external_proxy}, I extend Assumptions \ref{ass:proxy} and \ref{ass:score_funct} to allow for $\hat{A}_i$ to be produced as a data-dependent estimate. I show that, in case $\hat{A}_i = \hat{f}(X_i)$ where $\hat{f}$ is learned in an independent sample $S_m$ and then applied to $S_n$, the same results in Theorems \ref{thm:ha_upper} and \ref{thm:ha_lower} apply conditional on $S_m$.

\subsection{Minimax Comparison Between CB and Augmented Rules}\label{sec:minimax_comp}
\begin{corollary}[Minimax dominance of $\hat{a}$-Augmented rules]
\label{cor:proxy_domination}
Define $\mathcal{P}(\sigma_0, \rho)$ as a class of data-generating processes that satisfy $\bar{\sigma}_{\tau|x}(P)\leq \sigma_0$ and $\mathrm{rMSE}(\hat{A}_i)\leq \rho\leq 1/(2\kappa)$. 
Under Assumptions \ref{ass:kt18}, \ref{ass:proxy}, if
\begin{equation}
C_3\sigma_0 \geq \frac{M}{k}\frac{C_1\sqrt{v^\theta_{x,\hat{a}}} - C_2\sqrt{v^\theta_x}}{\sqrt{n}} + M \kappa L_s \rho
\end{equation}
then, for any $\mathcal{G}^\theta_x$ and $\mathcal{G}_{x,\hat{a}}^\theta$ that satisfy Assumption \ref{ass:score_funct},
\begin{equation}
\inf_{\{\hat{G}(X_i)\}}
\sup_{P \in \mathcal{P}(\sigma_0,\rho)}
R(\hat{G}_\theta (X_i))
 \ge 
\inf_{\{\hat{G}(X_i,\hat{A}_i)\}}
\sup_{P \in \mathcal{P}(\sigma_0,\rho)}
R(\hat{G}_\theta(X_i,\hat{A}_i))  
\end{equation}
\end{corollary}

Corollary \ref{cor:proxy_domination} shows that, when latent heterogeneity in treatment effects exceeds the sum of (i) the increase in policy space complexity due to adding $\hat{A}_i$ to the decision problem, and (ii) the probability of disagreement between the oracles with full and partial information of $A_i$ rescaled by $M$, then, accounting for unobserved heterogeneity through $\hat{A}_i$ when learning the optimal policy minimax-dominates ignoring it.
The result follows by combining Theorem \ref{thm:cb_lower} and Theorem \ref{thm:ha_upper} on the common class \(\mathcal P(\sigma_0,\rho)\). Indeed, the lower-bound construction used in the proof of Theorem \ref{thm:cb_lower} can be extended to \(\mathcal P(\sigma_0,\rho)\) by augmenting it with any proxy process satisfying Assumption \ref{ass:proxy} and \(\mathrm{rMSE}(\hat A_i)\le \rho\). Since Covariate-Based rules do not depend on \(\hat A_i\), this extension leaves the regret unchanged, and hence
\begin{equation}
\inf_{\{\hat G(X_i)\}}\sup_{P\in\mathcal P(\sigma_0,\rho)}R(\hat G_\theta(X_i))
\ge
C_2\frac{M}{k}\sqrt{\frac{v_x^\theta}{n}}+C_3\sigma_0.
\end{equation}
On the other hand, by Theorem \ref{thm:ha_upper}, for every \(P\in\mathcal P(\sigma_0,\rho)\),
\begin{equation}
R(\hat G_\theta(X_i,\hat A_i))
\le
C_1\frac{M}{k}\sqrt{\frac{v_{x,\hat a}^\theta}{n}}+M\kappa L_s\rho,
\end{equation}
and therefore
\begin{equation}
\inf_{\{\hat G(X_i,\hat A_i)\}}\sup_{P\in\mathcal P(\sigma_0,\rho)}R(\hat G_\theta(X_i,\hat A_i))
\le
C_1\frac{M}{k}\sqrt{\frac{v_{x,\hat a}^\theta}{n}}+M\kappa L_s\rho.
\end{equation}
The condition above then implies the claimed minimax ordering.

\section{Targeted Data Collections for Better Policies} \label{sec:data_collection}

In this section, I leverage the regret bounds derived in section \ref{sec:regret_bounds} to study how a policymaker should design data collection before learning policies. I consider the precision of the proxy $\hat{A}_i$ and the available sample size for learning policies as the outcome of ex ante design choices. On the one hand, collecting richer information on the latent factor, for instance by administering longitudinal surveys, collecting repeated measurements, or increasing training sample size for statistical models can improve the precision of $\hat{A}_i$. On the other hand, these same resources could be used to increase the sample size available for learning policies, for instance by running a larger field experiment, or acquiring a larger observational dataset. This creates a resource allocation problem between two competing objectives: reducing the measurement error in the proxy and reducing the statistical error in the estimated policy.

To formalize this problem, I introduce an information index $t \in \mathcal{T}$ that maps into the rMSE of $\hat{A}_i$. Higher values of $t$ correspond to richer information and therefore to more precise measurements of $A_i$.

\begin{assumption}[Information Index] \label{ass:time}
There exists an information index $t \in \mathcal T$ and a non-increasing function
$h:\mathcal T \to \mathbb R_+$ such that, for every $t \in \mathcal T$ and every
$P\in \mathcal{P}$,
\begin{equation}
\operatorname{rMSE}(\hat A_i(t)) \le h(t).
\end{equation}
\end{assumption}

Assumption \ref{ass:time} requires that there exists a function that maps the information index into the rMSE of $\hat{A}_i$ and that such function is non-increasing. Define $\hat{A}_i(t)$ as the measurement or estimate of $A_i$ under information $t$.

\begin{example}\label{ex:repeated_measure}
Suppose the policymaker cannot observe $A_i$ directly, but can collect $t \in \mathbb{N}$ independent noisy measurements of it: $M_{ij} = A_i + U_{ij},\ j=1,\dots,t$, where the measurement errors satisfy Assumption \ref{ass:proxy} and are conditionally independent across $j$ given $X_i$. Define the proxy as the sample average of the $t$ repeated measurements:
\begin{equation}
    \hat{A}_i(t):=\frac{1}{t}\sum_{j=1}^t M_{ij} .
\end{equation}
Define the conditional bias and variance of each measurement as:
\begin{equation}
    b_P(X_i):=\mathbb{E}_P[U_{ij}\mid X_i],
    \qquad
    \sigma^2_{U,P}(X_i):=\mathbb{V}_P(U_{ij}\mid X_i).
\end{equation}
If there exist constants $b_0,m_0>0$ such that, for every $P\in\mathcal P$,
\begin{equation}
    \sqrt{\mathbb{E}_P[b_P(X_i)^2]}\le b_0
    \qquad \text{and} \qquad
    \sqrt{\mathbb{E}_P[\sigma^2_{U,P}(X_i)]}\le m_0,
\end{equation}
then Assumption \ref{ass:time} is satisfied with the envelope:
\begin{equation}
    h(t)=b_0+\frac{m_0}{\sqrt{t}}.
\end{equation}
The formal proof is reported in Appendix \ref{app:examples}.
\end{example}

I now define the policymaker's design problem. The policymaker jointly chooses the information level $t$ and the sample size $n$ before learning the policy. The objective is to minimize worst-case regret subject to a finite budget. The policy can either ignore the proxy and rely only on observed covariates, or incorporate the proxy estimated at information level $t$.

\begin{definition}\label{def:baseline_collection}
Let $c(t,n)$ denote the total cost of collecting information level $t$ and sample size $n$. Consider the design problem of a policymaker that needs to decide the information level $t$ and the sample size $n$ under a budget constraint before learning the optimal policy $\hat{G}_\theta(X_i,\hat{A}_i(t))$:
\begin{align}
    & \min_{t \in \mathcal{T}, n\in \mathbb{N}} \left\{ \sup_{P \in \mathcal{P}(\sigma_0,h(t))}{R}(t, n) \right\} \\
    & \text{s.to:}\ c(t,n)\leq B_0
\end{align}
where $R(t,n):= \mathbb{E}_{P^n}[W(G^*_\theta(X_i,A_i)) - W(\hat{G}_\theta(Z_i(t)))]$ for $Z_i(t) \in \{X_i, (X_i,\hat{A}_i(t))\}$, and $\mathcal{P}(\sigma_0,h(t))$ is defined as in Corollary \ref{cor:proxy_domination} with $\rho = h(t)$.
\end{definition}

Definition \ref{def:baseline_collection} makes explicit that the policymaker faces two margins of choice. The first concerns whether to use a proxy for the latent factor at all, and if so with what level of precision. The second concerns how many observations to collect for learning the optimal policy. The budget constraint captures the idea that improving one dimension necessarily crowds out investment in the other.

In general, the results in section \ref{sec:regret_bounds} do not rewrite the minimax problem in Definition \ref{def:baseline_collection} exactly. Rather, they provide regret bounds that can be used to derive sufficient conditions for minimax dominance across feasible designs. 

Let $n(B_0,t):= \max\{n \in \mathbb N : c(t,n) \leq B_0\}$ denote the feasible sample size at budget level $B_0$ and information index $t$. 
The feasible sample size under the Covariate-Based design is $n_{CB}=n(B_0, 0)$, while the feasible sample size under the augmented design with information level $t$ is $n_A(t)=n(B_0, t)$. Given the results in Theorems \ref{thm:cb_upper}--\ref{thm:ha_lower}, define:
\begin{align}
    \overline{R}_{CB}(B_0) 
    &:= C_1 \frac{M}{k}\sqrt{\frac{v_x^\theta}{n_{CB}}} + \sigma_0, \\
    \underline{R}_{CB}(B_0) 
    &:= C_2 \frac{M}{k}\sqrt{\frac{v_x^\theta}{n_{CB}}} + C_3 \sigma_0, \\
    \overline{R}_{A}(t,B_0) 
    &:= C_1 \frac{M}{k}\sqrt{\frac{v_{x,\hat{a}}^\theta}{n_A(t)}} + M\kappa L_sh(t), \\
    \underline{R}_{A}(t,B_0) 
    &:= C_4 \frac{M}{k}\sqrt{\frac{v_{x,\hat{a}}^\theta}{n_A(t)}} + C_5 M\kappa h(t).
\end{align}
By Theorems \ref{thm:cb_upper} and \ref{thm:ha_upper}, $\overline{R}_{CB}(B_0)$ and $\overline{R}_{A}(t,B_0)$ are valid upper bounds on the minimax regret of the feasible Covariate-Based and augmented designs. By Theorems \ref{thm:cb_lower} and \ref{thm:ha_lower}, together with the same extension argument used in Corollary \ref{cor:proxy_domination}, $\underline{R}_{CB}(B_0)$ and $\underline{R}_{A}(t,B_0)$ are valid lower bounds.

Let $V_{CB}(B_0)$ denote the minimax regret of the feasible Covariate-Based design, and let $V_A(t,B_0)$ denote the minimax regret of the feasible augmented design with information level $t$. 

Assume the policymaker must have some prior information on the severity of the approximation error incurred by ignoring latent heterogeneity.

\begin{assumption}[Prior on $\bar{\sigma}_{\tau|x}$]\label{ass:prior_sigma}
    Assume the policymaker has some prior knowledge on the conditional variance of the treatment effect $\bar{\sigma}_{\tau|x}$:
    \begin{equation}
        \bar{\sigma}_{\tau|x} \leq \sigma_0
    \end{equation}
\end{assumption}

Assumption \ref{ass:prior_sigma} does not require point identification of the unexplained heterogeneity in treatment effects. It only requires an upper bound that can be interpreted as prior or contextual knowledge about the empirical relevance of latent heterogeneity.

The next proposition uses these bounds to provide sufficient conditions for minimax dominance across feasible designs.

\begin{proposition}[Sufficient Conditions for Minimax Dominance] \label{prop:minimax_design}
Under the Assumptions \ref{ass:kt18} to \ref{ass:prior_sigma}, the following statements hold:
\begin{enumerate}
    \item If, $\forall t \in \mathcal{T}$,
    \begin{equation}
        \overline{R}_{CB}(B_0) < \underline{R}_{A}(t,B_0),
    \end{equation}
    then
    \begin{equation}
        V_{CB}(B_0) < V_A(t,B_0)\qquad \forall t\in\mathcal{T},
    \end{equation}
    so the Covariate-Based design minimax-dominates every augmented design. As a consequence, 
    \begin{equation}
        (t^*,n^*) = \left( 0, n(B_0,0) \right)
    \end{equation}

    \item If there exists $t^*\in\mathcal{T}$ such that
    \begin{equation}
        \overline{R}_{A}(t^*,B_0)
        <
        \min \left\{\underline{R}_{CB}(B_0), \inf_{t\in\mathcal{T}\setminus\{t^*\}}\underline{R}_{A}(t,B_0),\right\}
    \end{equation}
    then
    \begin{equation}
        V_A(t^*,B_0) < V_A(t,B_0)\ \forall t\neq t^*\qquad \wedge \qquad V_A(t^*,B_0) < V_{CB}(B_0)
    \end{equation}
    so the augmented design with information level $t^*$ is the minimax-dominant design with respect to the CB rule, and within the class of augmented rules. As a consequence,
    \begin{equation}
        (t^*,n^*) = \left(t^*, n(B_0,t^*) \right)
    \end{equation}
\end{enumerate}
\end{proposition}

Proposition \ref{prop:minimax_design} provides sufficient conditions under which the regret bounds from Theorems \ref{thm:cb_upper}--\ref{thm:ha_lower} are informative enough to certify minimax dominance across feasible designs. Part 1 delivers a pairwise comparison between the Covariate-Based design and any augmented design indexed by $t$. Part 2 can then be used to rank augmented designs among themselves and identify a minimax-dominant information level whenever the comparison with the Covariate-Based design does not certify dominance in its favor. 
The formal proof is reported in Appendix \ref{app:formal_proofs}.

\section{Empirical Application}\label{sec:emp_application}
In this section, I introduce two procedures to (i) rank policy rules that ignore or incorporate unobserved heterogeneity and (ii) estimate the optimal allocation of budget between the tasks of measuring (or estimating) latent factors and estimating policies. I conduct three empirical exercises and deliver new insights on the data from \cite{hussam_2022}. 

The authors study the effect of providing a cash grant to micro-entrepreneurs on their profits with a randomized controlled trial in rural India. 
They introduce a new proxy to measure entrepreneurs' business skills, an unobserved dimension identified from previous literature as policy-relevant for targeting interventions apt to stimulate economic development. 
This proxy is based on the ranking that groups of five entrepreneurs give each other across different outcomes. 
The authors name this proxy \textit{community rankings} and claim as their main result that it can help target high-growth micro-entrepreneurs. 

The study by \cite{hussam_2022} provides a good setting of application for two main reasons. 
First, the applied research question, whether targeting based on a proxy of a policy-relevant unobserved characteristic is welfare-improving, is strongly aligned with the theoretical investigation of the present paper.
Second, the way the proxy is measured, through the average of five repeated measurements, allows me to study how the performance of policy recommendations varies with the precision of community rankings, and estimate the optimal allocation of budget between larger experiments and higher number of measurements.

In the first exercise, I confirm qualitatively the main result from \cite{hussam_2022} and provide new estimates for the magnitude of the welfare gains. I show that targeting resources along the values of community rankings increases average welfare by $5\%$, and reduces by two thirds the probability of producing welfare losses (harm rate for later reference) as compared to scaling up the intervention by random assignment. 
This gain reduces to $3\%$ welfare increase and half harm rate reduction, when compared to covariate based rules. 

In the second exercise, I leverage the fact that the proxy was based on the average measurement of five separate rankers to show that, keeping sample size fixed, the gains of targeting based on community ranking increase with the number of rankers.

In the third exercise, I impose a budget constraint and estimate the optimal number of measurements and sample sizes for different budgets. 
As new insights, I show that 
(i) even for limited budgets, it is never optimal to ignore the heterogeneity induced by business skills; 
(ii) when budget is limited, it is optimal to collect fewer measurements, in favor of a larger sample size; 
(iii) for high budgets, it is optimal to collect as many measurements as possible.

\subsection{Summary of the Experimental Design}
The trial was conducted in the city of Amravati, India, between 2016 and 2018. 
It was designed to assess whether local community members possess predictive information about heterogeneity in entrepreneurial returns and can be useful to improve the targeting of cash grants. 

The sample consists of 1,345 micro-entrepreneurs operating informal businesses in retail and services. 
First, participants were assigned to peer groups of five or six based on geographic proximity. 
Within these groups, individuals were asked to rank their peers on future business outcomes, including future profits and marginal returns to capital. 
The main measure of community ranking used in the paper is the average fraction of peers who ranked a given entrepreneur in the top quartile across the different outcomes for which rankings were elicited. 
One-third of the sample was then randomly assigned to receive an unconditional cash grant of 6,000 INR (roughly \$100).

The available data include a set of characteristics collected at baseline and after the treatment.  
I consider as outcome variable the profits realized 2 months after the intervention. 

\subsection{Ranking Policy Classes}
In this section, I rank CB and $\hat{a}$-Augmented rules and quantify the welfare gains from incorporating community rankings into treatment assignment.       

\paragraph{Policy Classes.} I consider the following Covariate-Based rule:
\begin{equation}
    G(X_i) = \mathds{1}\left\{ X_{i,1} \geq t_1\ \&\ X_{i,2} \geq t_2  \right\}
\end{equation}
where $X_{i,1}$ is age and $X_{i,2}$ is education in years. Age and education are both identified as policy-relevant dimensions by \cite{hussam_2022} and previous literature. 
The $\hat{a}$-CB rule is then defined as:
\begin{equation}
    G(X_i, \hat{A}_i)= \mathds{1}\left\{ X_{i,1} \geq t_1\ \&\ X_{i,2} \geq t_2\ \&\ \hat{A}_i > t_{\hat{a}} \right\}
\end{equation}
where $\hat{A}_i$ is community ranking.

Finally, I also consider a benchmark random rule $G_\text{rand}$ that assigns the treatment at random. 
To evaluate the performance of each rule, I first randomly split the sample into an estimating and test set; then, I use the estimating set to estimate the rules $G(X_i)$, $G(X_i,\hat{A}_i)$ that solve the respective maximization problems; finally, I compute the empirical welfare generated by each estimated rule in the test sample. 
I leverage the randomness of the sample split to recover the distribution of out-of-sample empirical welfare over $B=2000$ different draws of the estimating and test set data. The sample splitting procedure is illustrated in Figure \ref{fig:sample_split}. I illustrate the evaluation algorithm in Algorithm \ref{alg:welf_evaluation}.

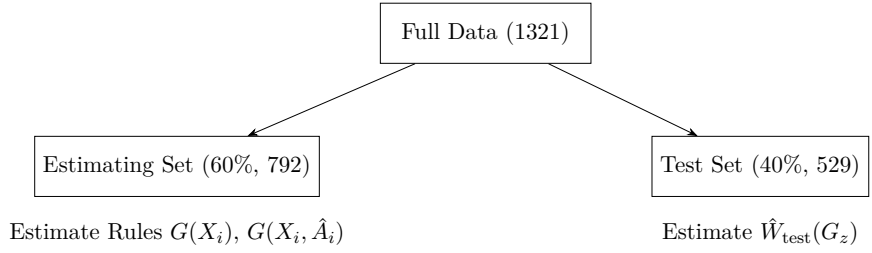
\begin{figure} 
\caption{Sample Splitting} 
\label{fig:sample_split}
\centering
\scalebox{0.8}{
\begin{tikzpicture}[node distance=1.2cm and 1cm, every node/.style={font=\small}] 
  \node[draw, rectangle, minimum width=3.5cm, minimum height=1cm] (full) {Full Data (1321)};
  \node[draw, rectangle, minimum width=3cm, minimum height=1cm, below left=of full] (sample) {Estimating Set (60\%, 792)};
  \node[draw, rectangle, minimum width=3cm, minimum height=1cm, below right=of full] (test) {Test Set (40\%, 529)};

  \draw[-{Stealth}] (full) -- (sample);
  \draw[-{Stealth}] (full) -- (test);
  
  \node[below=0.2cm of sample] {Estimate Rules $G(X_i)$, $G(X_i,\hat{A}_i)$};
  \node[below=0.2cm of test] {Estimate $\hat{W}_{\text{test}}(G_z)$};
\end{tikzpicture}}
\begin{minipage}{1\linewidth}
    \footnotesize
    \textit{\textbf{Notes:}} This figure illustrates the sample splitting procedure specifying the relative and absolute size of each split. All shares are relative to the total sample. 
\end{minipage}
\end{figure}

The test set empirical welfare of a given rule $G_z$ is computed as:

\begin{equation} \label{eq:test_welfare}
    \hat{W}_{\text{test}}(G_z) = \frac{1}{529} \sum_{i\in S_\text{test}} \left[ \frac{Y_i D_i}{1/3} \cdot \mathbf{1}\{i \in G_z\} + \frac{Y_i (1 - D_i)}{2/3} \cdot \mathbf{1}\{i \notin G_z\} \right] 
\end{equation}

where $Y_i$ denotes profits the micro-entrepreneur made in the 60 days following the intervention.

In Figure \ref{fig:welfare_ranking}, I report the cumulative distribution of welfare over the different draws of the estimating and test sets. 
In column 1 of Table \ref{tab:welfare_gains}, I report the empirical cdf of welfare evaluated at the status quo, the harm rate. It measures the probability that a given rule generates a welfare lower than the status quo. 
In columns $(2)-(4)$ of Table \ref{tab:welfare_gains}, I report the average pairwise difference in test welfare between different rules.

First, all non-random rules dominate the random rule. Therefore, if a government were to scale up this intervention, scaling it without targeting would not be optimal. 
Second, the CB rule is stochastically dominated by augmented rules. Therefore, as claimed in \cite{hussam_2022}, using community ranking as a targeting variable produces a welfare gain. In particular, $\hat{a}$-CB rules achieve an average welfare $247\$$ ($5\%$) higher than random rules and $182\$$ ($4\%$) higher than CB rules.
Finally, targeting using community rankings reduces the harm rate by a half compared to random rules and a third compared to CB rules. 
This means that, had the policymaker scaled up the cash transfer intervention learning the optimal $\hat{a}$-CB rule from a sample of the size of the training set, the probability of that being harmful over the distribution of the estimating sample is reduced by a half (third) compared to random (CB) rules.

\begin{figure}[h!]
    \centering
    \caption{Welfare Gains by Policy Rule}\label{fig:welfare_ranking}
    \scalebox{1.6}{
    \includegraphics[width=0.4\linewidth]{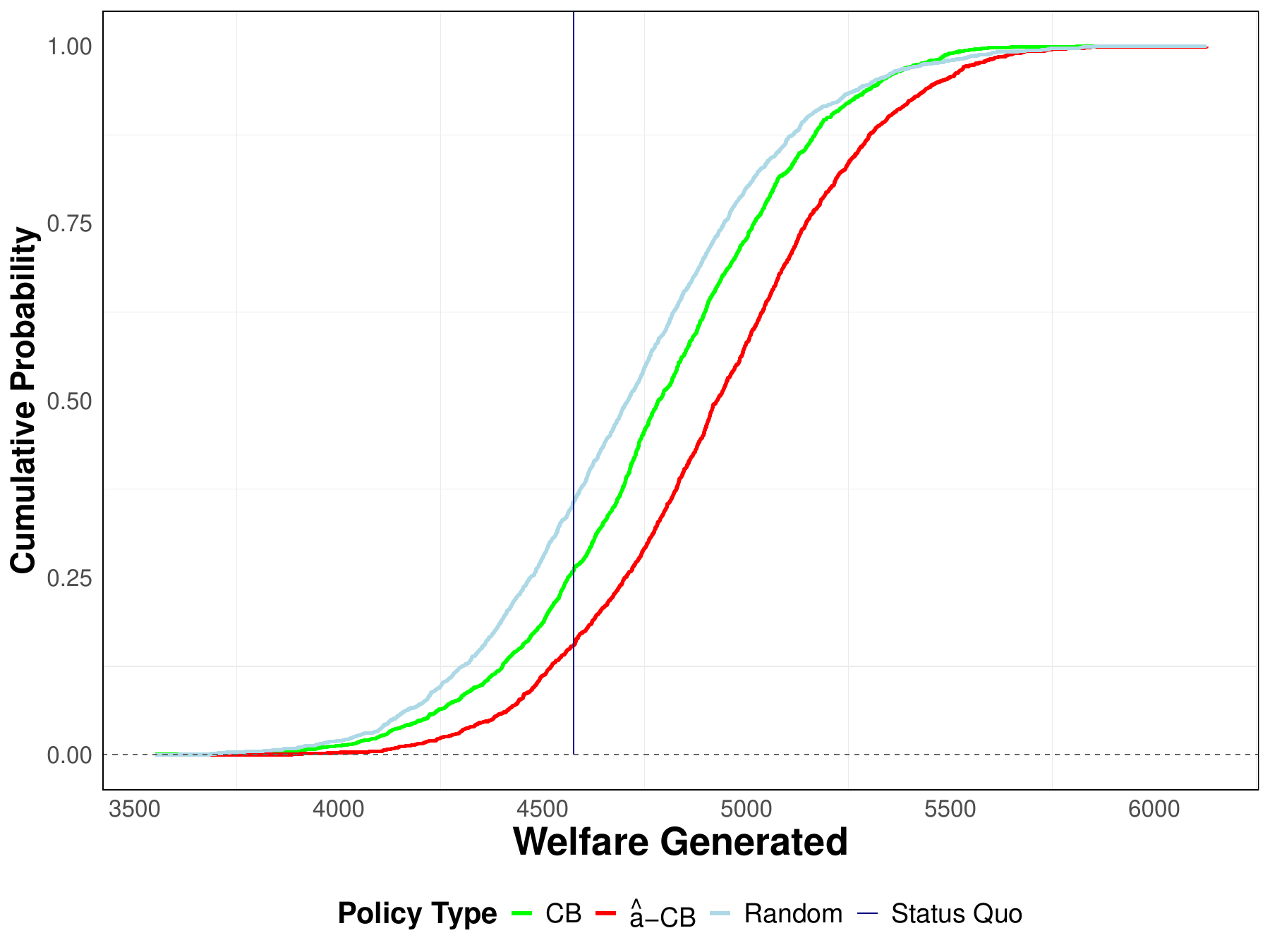}}
    \begin{minipage}{1\linewidth}
    \footnotesize 
    \textit{\textbf{Notes:}} This figure illustrates the empirical cumulative distribution of welfare generated by $\hat{G}_\text{rand}$ (light blue), $\hat{G}(X_i)$ (green), $\hat{G}(X_i,\hat{A}_i)$ (red) over 2000 draws of the estimating and test set sample split. I indicate the status quo (average outcome for untreated units) with the solid navy line. Algorithm \ref{alg:welf_evaluation} illustrates the procedure, Figure \ref{fig:sample_split}, and Equation \ref{eq:test_welfare} define the welfare measure.
    \end{minipage}
\end{figure}

\begin{table}[htpb!]
\centering
\begin{threeparttable}
\caption{Welfare Gains by Policy Rule}
\label{tab:welfare_gains}
\begin{tabular}{lrrrr}
\toprule
\textbf{Policy Rule} & Harm Rate & Rand. & CB & $\hat{a}$-CB \\
\midrule
& (1) & (2) & (3) & (4) \\
\midrule
Status Quo & - & +171\$\ (+4\%) & +245\$\ (+5\%) & +384\$\ (+8\%) \\
Rand. & 0.32 & - & +73\$\ (+2\%) & +213\$\ (+5\%) \\
CB & 0.22 & - & - & +139\$\ (+3\%) \\
$\hat{a}$-CB & 0.14 & - & - & - \\
\midrule
Status Quo & 4,540\$ & - & - & - \\
\bottomrule
\end{tabular}
\begin{tablenotes}[flushleft]
\footnotesize
\item \textit{\textbf{Notes:}} Each cell reports the mean welfare gain of the column policy over the row policy (in \$, with percentage relative to the status quo welfare level), averaged across $B=2{,}000$ sample-splitting replications. Harm Rate denotes the share of replications in which the policy produces lower average welfare than the status quo. The bottom row reports the mean status quo welfare level (in \$).
\end{tablenotes}
\end{threeparttable}
\end{table}

\begin{algorithm}[H] 
\caption{Welfare Evaluation}\label{alg:welf_evaluation}
\begin{algorithmic}[1] 
\For{$b = 1$ to $B= 2000$}
  \State Set random seed to $b$.
  \State Random split: $S_n = S^b_\text{est} \cup S^b_\text{test}$
  \State \textbf{Compute Rules:}
  \State Estimate $G(X_i)$ and $G(X_i,\hat{A}_i)$ using $S^b_{\text{est}}$.
  \State \textbf{Evaluate Rules:}
  \State Estimate $\hat{W}^b_{\text{test}}(\hat{G}_\text{r})$, $\hat{W}^b_{\text{test}}(\hat{G}(X_i))$, $\hat{W}^b_{\text{test}}(\hat{G}(X_i,\hat{A}_i))$.
\EndFor
\end{algorithmic}
\end{algorithm}

\subsection{Evidence of Decay of Performance}

In this section, I provide empirical evidence for the theoretical prediction embedded in Theorem \ref{thm:ha_upper}: welfare gains from $\hat{a}$-CB rules decrease as proxy noise increases.
Recall that community rankings are defined as the average fraction of peers who rank a given entrepreneur in the top quartile, elicited from four or five separate rankers.\footnote{Only 37 entrepreneurs have five rankers.}
This feature of the experimental design allows me to vary proxy precision by restricting the number of rankers used to construct it.

Fixing the sample size to the full dataset, I define $\hat{a}_j$ as the community ranking proxy constructed from $j \in \{1,\ldots,5\}$ randomly selected rankers. Higher values of $j$ correspond to more precise measurements of the latent business skill, with $\hat{a}_5$ coinciding with the full proxy analyzed in the previous subsection.\footnote{Refer to Example \ref{ex:repeated_measure} for a formal justification.}
In Figure \ref{fig:original_vs_tmeasures}, I compare the original measure with $\hat{a}_j$ and show that for $j=5$ the two measures coincide, while as $j$ decreases $\hat{a}_j$ gets scattered around the original proxy.

Table \ref{tab:welfare_gains_measures} reports, for each value of $j$, the average welfare gain of the $\hat{a}_j$-CB rule relative to three benchmarks: the status quo, i.e., treating no one (column 1); random assignment (column 2); and the CB rule (column 3). 
To avoid ranker-specific effects, the average is computed over the sample splits $B=2000$ and over $R=30$ random selections of $j$ rankers.

The welfare gain of $\hat{a}_j$-CB over random assignment and CB rules is positive and increases monotonically in $j$ for $j\in [1,4]$.
This pattern mimicks the upper bound in Theorem \ref{thm:ha_upper}: as $\operatorname{rMSE}(\hat{A}_j)$ shrinks, the noise-related term in the $\hat{a}$-CB regret bound falls, narrowing the gap relative to the oracle and widening the welfare advantage over rules that ignore latent heterogeneity altogether.
One puzzling result is that welfare gains slightly decrease at $j=5$. This pattern may be explained by the lack of statistical power due to the small size of the sample, expecially considering that only 37 entrepreneurs in the sample have 5 separate non-self rankers.

\begin{table}[htpb!]
\centering
\begin{threeparttable}
\caption{Welfare Gains of $\hat{a}$-CB by Number of Measurements}
\label{tab:welfare_gains_measures}
\begin{tabular}{lrrr}
\toprule
\textbf{Measure} & vs.\ Status Quo & vs.\ Random & vs.\ CB \\
\midrule
& (1) & (2) & (3) \\
\midrule
$\hat{a}_1$ & +304\$\ (+7\%) & +113\$\ (+2\%) & +52\$\ (+1\%) \\
$\hat{a}_2$ & +346\$\ (+8\%) & +158\$\ (+3\%) & +94\$\ (+2\%) \\
$\hat{a}_3$ & +362\$\ (+8\%) & +171\$\ (+4\%) & +110\$\ (+2\%) \\
$\hat{a}_4$ & +407\$\ (+9\%) & +218\$\ (+5\%) & +155\$\ (+3\%) \\
$\hat{a}_5$ & +392\$\ (+9\%) & +203\$\ (+4\%) & +140\$\ (+3\%) \\
\bottomrule
\end{tabular}
\begin{tablenotes}[flushleft]
\footnotesize
\item \textit{\textbf{Notes:}} Each row corresponds to a different information source used to construct $\hat{a}$. Each cell reports the mean welfare gain of $\hat{a}$-CB over the comparison policy (in \$, with percentage relative to the mean status quo welfare level), averaged across $B=2{,}000$ sample-splitting replications, and $R = 30$ random selection of rankers.
\end{tablenotes}
\end{threeparttable}
\end{table}

\subsection{Estimating Optimal Designs}

In this section, I consider the problem of estimating the optimal data collection plan. 
The key design margin in this setting is the number of peer rankings used to construct the proxy for business skill. 
I use this feature of the data to study how welfare changes as the policymaker trades off the amount of information used to measure the latent trait against the sample size used to learn the optimal policy.

Formally, let $t \in \{0,1,2,3,4,5\}$ denote the number of non-self rankers used to construct the proxy. The case $t=0$ corresponds to a design in which no ranking information is collected and the policymaker relies only on Covariate-Based rules. For $t>0$, I randomly draw $t$ rankers among those available for each entrepreneur, and compute the average ranking across the selected rankers. Higher values of $t$ correspond to richer information and therefore to more precise measurements of the latent trait.

To introduce the budget constraint, suppose that collecting one observation for policy learning costs $c_n$, while each additional ranking used to construct the proxy costs $c_t$ for each unit. Then, for a given budget $B_0$, the feasible sample size satisfies
\begin{equation}
    n(t,B_0) = \left\lfloor \frac{B_0}{c_n + c_t t} \right\rfloor.
\end{equation}
Therefore, increasing $t$ improves the precision of the proxy but reduces the number of observations that can be used to learn the policy. 

I evaluate this trade-off over a grid of budgets $B_0 \in \{600,800,...,2000\}$, setting $c_n=0.75$ and $c_t=0.25$. For each budget and each value of $t$, I estimate the welfare generated by the feasible design using repeated sample splitting. 
At the beginning of each repetition, I draw a common test sample and a common training pool from the main analysis data. 
When $t=0$, I draw up to $n(t,B_0)$ observations from the training pool and estimate the Covariate-Based rectangular rule defined above. 
When $t>0$, I first generate a random proxy $\hat{A}_i(t)$ by selecting $t$ rankers for each entrepreneur. 
I then draw up to $n(t,B_0)$ observations from the resulting training pool. 
On this feasible sample, I estimate both the Covariate-Based rule and the augmented rule. 
As in the ranking exercise, I also consider a benchmark random rule $G_{\text{rand}}$. 
The algorithm is described formally in Algorithm \ref{alg:design}.

I then evaluate the out-of-sample welfare generated by each estimated rule in the corresponding test sample. Repeating this procedure over $B=200$ sample splits and, for each $t>0$, over $R=30$ random realizations of the proxy allows one to recover the average welfare associated with each feasible design. Finally, within each budget level, I define the optimal design as the one that yields the highest average welfare. This procedure allows me to trace the welfare frontier over feasible designs and to estimate how the optimal allocation between the number of measurements $t$ and the policy-learning sample size $n$ changes with the available budget. 

Table \ref{tab:design_problem} and Figure \ref{fig:optimal_n_t} report the main findings. Three results stand out.                                                                                                             
      
First, the optimal design always includes proxy measurements: $t^* \geq 2$ for every budget level considered.    
Even at the tightest budget (\$600), allocating resources to community rankings, despite reducing the policy-learning sample from 794 to 480 observations, yields a welfare gain of $\$100$ $(+2\%)$ over the CB rule computed with maximum sample.
Ignoring latent heterogeneity is suboptimal even when measurement is costly.                                                                                                        
                                                                                                                 
Second, the optimal number of rankers increases with the budget. At low budgets $(\$600-\$800)$, $t^*=2$: the marginal cost of additional rankers crowds out too much sample size, so fewer measurements and a larger sample are preferred. 
From $\$1,000$ onwards, the constraint relaxes and $t^*=4$ becomes optimal, combining higher proxy precision with a feasible sample size.                                                                         
                                                                                                                 
Third, the design saturates: from $\$1,400$, $n^*$ reaches the sample cap (793 observations) and further budget increases yield no additional welfare gain. The welfare frontier flattens, with gains stabilizing at $+\$179$ $(+4\%)$ over the CB benchmark.                  

In Figures \ref{fig:optimal_n_t_c1} and \ref{fig:optimal_n_t_c2} we report the results for different cost functions. 
In Figure \ref{fig:optimal_n_t_c1} we consider the case where the cost of collecting one more measurement is higher than the cost of collecting one more experimental unit. In this case, it is optimal to colelct two measurements, for any budget.
In Figure \ref{fig:optimal_n_t_c2} we consider the case where the two marginal costs are equal. In this case, the conclusions are closer to the main specification.

\begin{figure}[h!]
    \centering
    \caption{Optimal Collection Plans by Budget Levels}
    \label{fig:optimal_n_t}
    \includegraphics[width=0.49\linewidth]{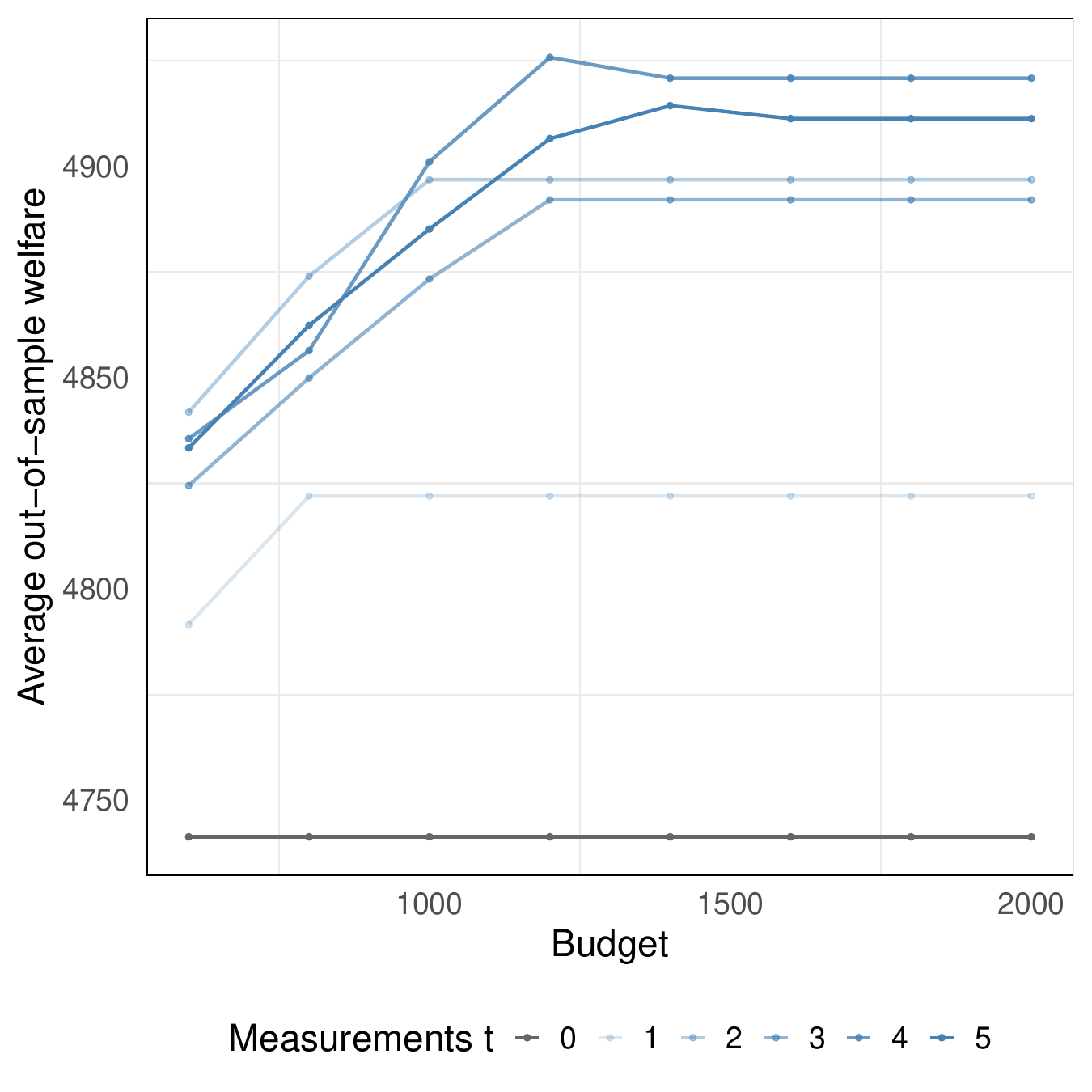}
    \includegraphics[width=0.49\linewidth]{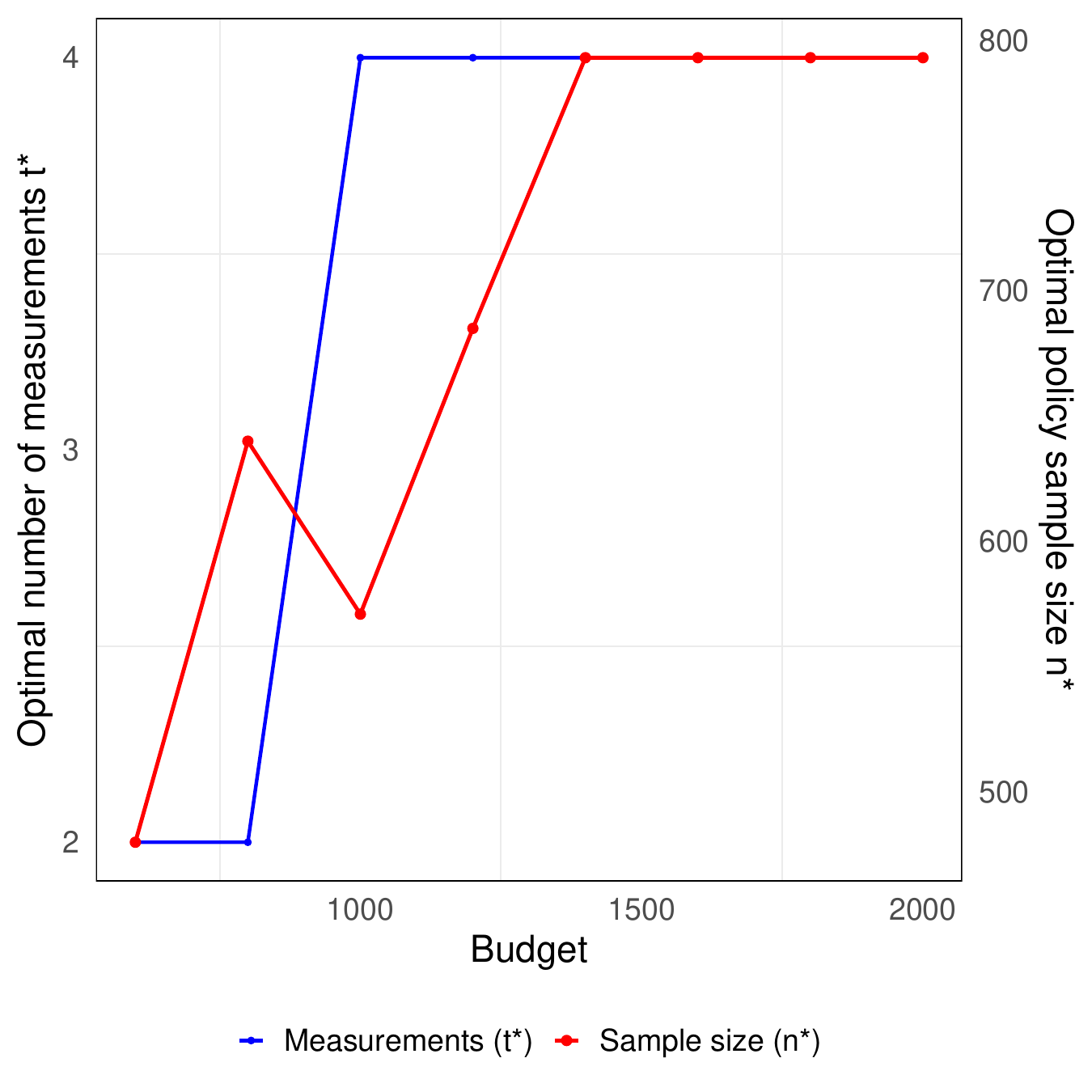}
    \begin{minipage}{1\linewidth}
    \footnotesize
    \textit{\textbf{Notes:}} This figure illustrates the performance of feasible data collection plans across different budget levels. The left panel reports the average out-of-sample welfare generated by designs with different numbers of measurements $t \in \{0,1,...,5\}$. The case $t=0$ corresponds to Covariate-Based rules, while $t>0$ corresponds to $\hat{a}$-Augmented rules constructed using $t$ randomly selected rankers. The right panel reports the optimal design as a function of the budget, showing the optimal number of measurements $t^*$ (left axis) and the corresponding optimal policy-learning sample size $n^*$ (right axis). 
    All results are obtained using the design evaluation procedure described in Algorithm \ref{alg:design}.
    \end{minipage}
\end{figure}

\begin{algorithm}[h!]
\caption{Design Evaluation Under a Budget Constraint}\label{alg:design}
\begin{algorithmic}[1]
\For{$b = 1$ to $B=200$}
    \State Set random seed to $b$.
    \State Random split: $S_n = S^b_\text{est} \cup S^b_\text{test}$
    \For{$B_0 \in \{600,800,..., 2000\}$}
        \For{$t \in \{0,1,2,3,4,5\}$}
            \State Compute feasible sample size $n(t,B_0)=\left\lfloor \frac{B_0}{c_n+c_t t} \right\rfloor$.
            \If{$t=0$}
                \State Draw $n(0,B_0)$ observations from $S_{\mathrm{est}}$.
                \State Estimate $G(X_i)$.
                \State Estimate ${W}^{b}_{\text{test}}(\hat{G}(X_i))$ and ${W}^{b}_{\text{test}}(\hat{G}_{\text{rand}})$.
            \Else
                \For{$r = 1$ to $R=30$}
                    \State Set random seed to $100000 \cdot b + 1000 \cdot t + r$.
                    \State Estimate $\hat{A}_i^{(b,r)}(t)$ for $t$ randomly selected rankers.
                    \State Draw $n(t,B_0)$ observations from $S_{\mathrm{est}}$.
                    \State Estimate $G(X_i)$ and $G(X_i,\hat{A}_i^{(b,r)}(t))$.
                    \State Estimate ${W}^{(b,r)}_{\text{test}}(\hat{G}(X_i))$, ${W}^{(b,r)}_{\text{test}}(\hat{G}(X_i,\hat{A}_i(t)))$, and ${W}^{(b,r)}_{\text{test}}(\hat{G}_{\text{rand}})$.
                \EndFor
            \EndIf
        \EndFor
    \EndFor
\EndFor
\end{algorithmic}
\end{algorithm}

\begin{table}[htpb!]
\centering
\begin{threeparttable}
\caption{Optimal Collection Plans by Budget Levels}
\label{tab:design_problem}
\begin{tabular}{lrrrrrr}
\toprule
Budget & $t^*$ & $n^*$ & Welfare$(t^*)$ & $n_0$ & Welfare$(t=0)$ & Gain \\
\midrule
& (1) & (2) & (3) & (4) & (5) & (6) \\
\midrule
600\$ & 2 & 480 & 4,842\$ & 794 & 4,741\$ & +100\$\ (+2\%) \\
800\$ & 2 & 640 & 4,874\$ & 794 & 4,741\$ & +133\$\ (+3\%) \\
1,000\$ & 4 & 571 & 4,901\$ & 794 & 4,741\$ & +160\$\ (+3\%) \\
1,200\$ & 4 & 685 & 4,926\$ & 794 & 4,741\$ & +184\$\ (+4\%) \\
1,400\$ & 4 & 793 & 4,921\$ & 794 & 4,741\$ & +179\$\ (+4\%) \\
1,600\$ & 4 & 793 & 4,921\$ & 794 & 4,741\$ & +179\$\ (+4\%) \\
1,800\$ & 4 & 793 & 4,921\$ & 794 & 4,741\$ & +179\$\ (+4\%) \\
2,000\$ & 4 & 793 & 4,921\$ & 794 & 4,741\$ & +179\$\ (+4\%) \\
\bottomrule
\end{tabular}
\begin{tablenotes}[flushleft]
\footnotesize
\item \textit{\textbf{Notes:}} For each budget level, columns (1)--(3) report the optimal number of rankers $t^*$, the resulting feasible sample size $n^*$, and the average out-of-sample welfare achieved. Columns (4)--(5) report the sample size and welfare under the CB-only benchmark ($t=0$). Column (6) reports the mean welfare gain of the optimal design over CB-only (in \$, with percentage), averaged across $B=200$ sample-splitting replications with $R=30$ proxy draws each. Costs: \$0.75 per observation, \$0.25 per ranking.
\end{tablenotes}
\end{threeparttable}
\end{table}

\section{Conclusions}\label{sec:conclusion} 
Standard policy learning studies the performance of treatment assignment rules based on observable characteristics. 
A large body of empirical work has established that latent traits, such as ability, motivation, or business skills, are of first-order importance in understanding treatment effect heterogeneity. 
Incorporating these traits into assignment rules comes with two costs: (i) measurement error propagates into the welfare criterion and (ii) the complexity of the policy class increases.
                                                            
I study this trade-off formally deriving rate-sharp regret bounds for Covariate-Based and $\hat{a}$-Augmented rules, showing that the proxy's inclusion improves worst-case performance only when the treatment effect variation explained by the latent factor outweighs the combined costs of noise propagation and policy space complexity. 
A new definition of regret, relative to an oracle that directly observes $A_i$, provides a common benchmark that makes this derivation tractable and the comparison meaningful.

Moreover, I frame the allocation problem between improving measurement precision and enlarging the policy-learning sample. I derive the conditions that separate the two regimes, yielding a sufficient condition for the minimax optimal allocation of resources, and propose sample-splitting procedures to implement these findings empirically.

In an application to \cite{hussam_2022}, I show that incorporating community rankings improves average welfare by $4\%$ and halves the probability of generating welfare losses relative to Covariate-Based rules. 
Moreover, I show that ignoring latent heterogeneity is not optimal, even under tight budget constraints, and that the optimal number of rankers increases with the available budget.

\bibliography{bib}

\clearpage

\appendix

\renewcommand{\thetable}{A\arabic{table}}
\renewcommand{\thefigure}{A\arabic{figure}}
\renewcommand{\thelemma}{A\arabic{lemma}}
\renewcommand{\thecorollary}{A\arabic{corollary}}
\renewcommand{\thedefinition}{A\arabic{definition}}
\setcounter{table}{0}
\setcounter{figure}{0}
\setcounter{lemma}{0}
\setcounter{corollary}{0}
\setcounter{definition}{0}

\section{Additional Figures}

\begin{figure}[h!]
    \centering
    \caption{Original Measure vs t-Measurements}
    \label{fig:original_vs_tmeasures}
    \includegraphics[width=0.45\linewidth]{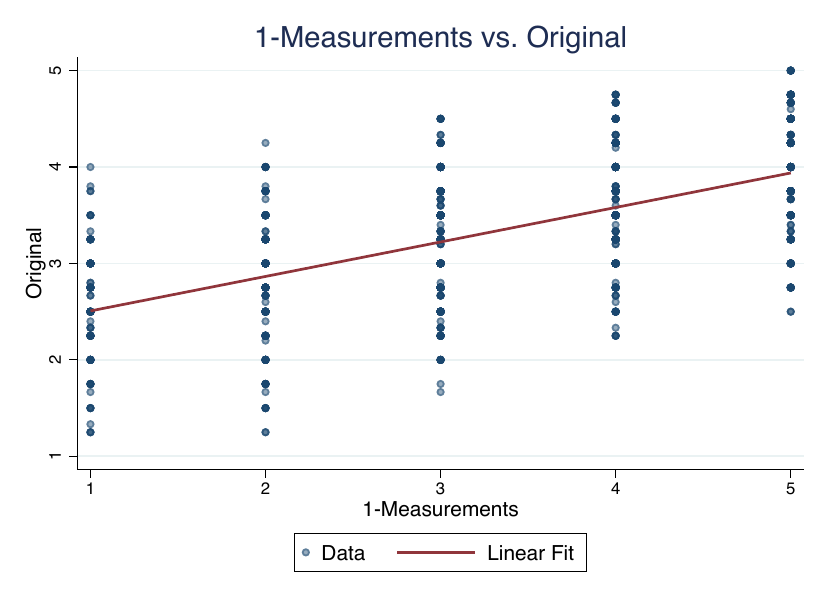}
    \includegraphics[width=0.45\linewidth]{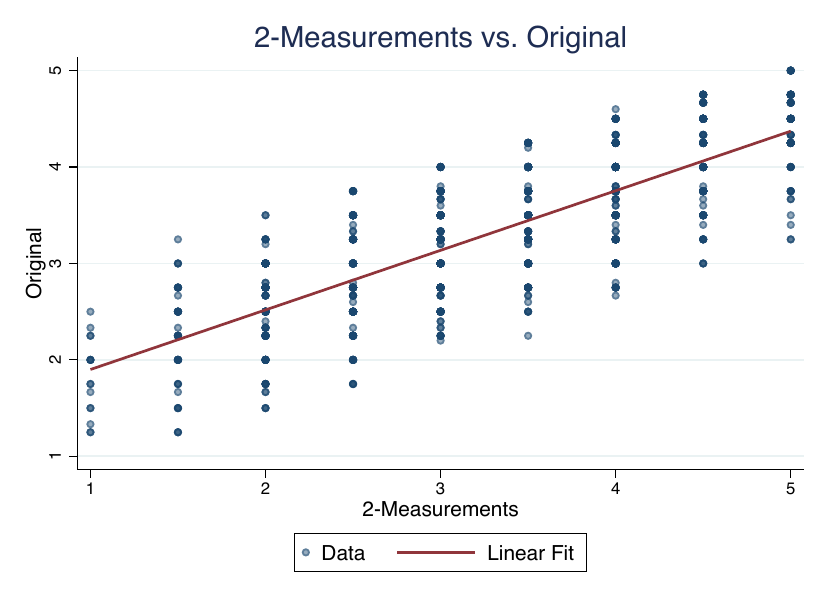}
    \includegraphics[width=0.45\linewidth]{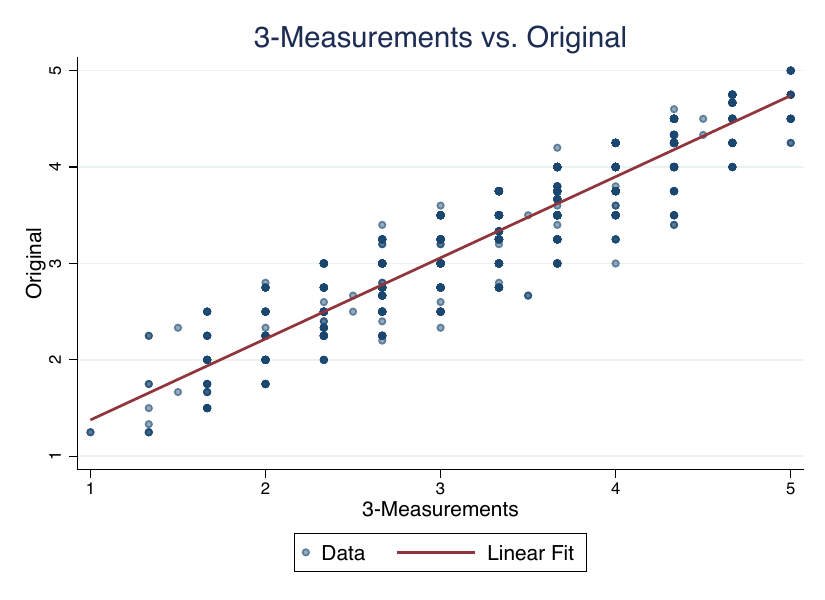}
    \includegraphics[width=0.45\linewidth]{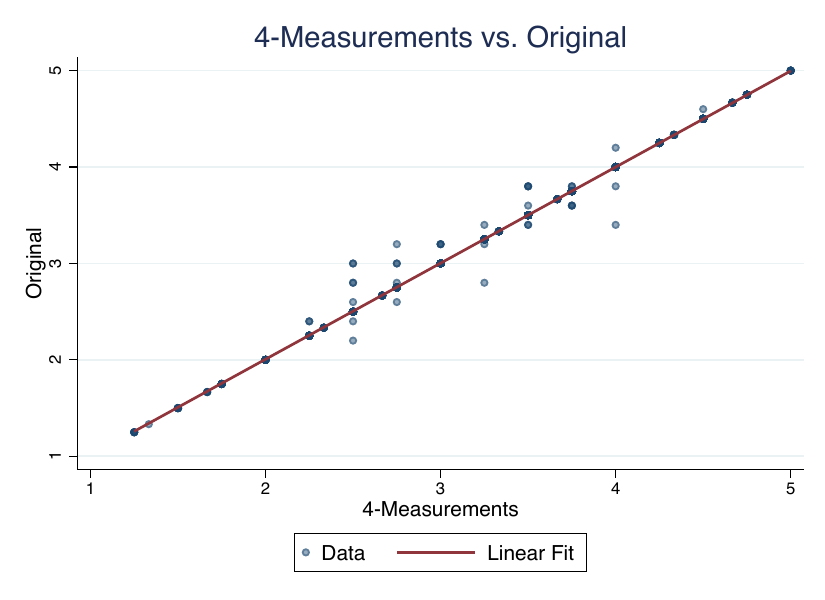}
    \includegraphics[width=0.45\linewidth]{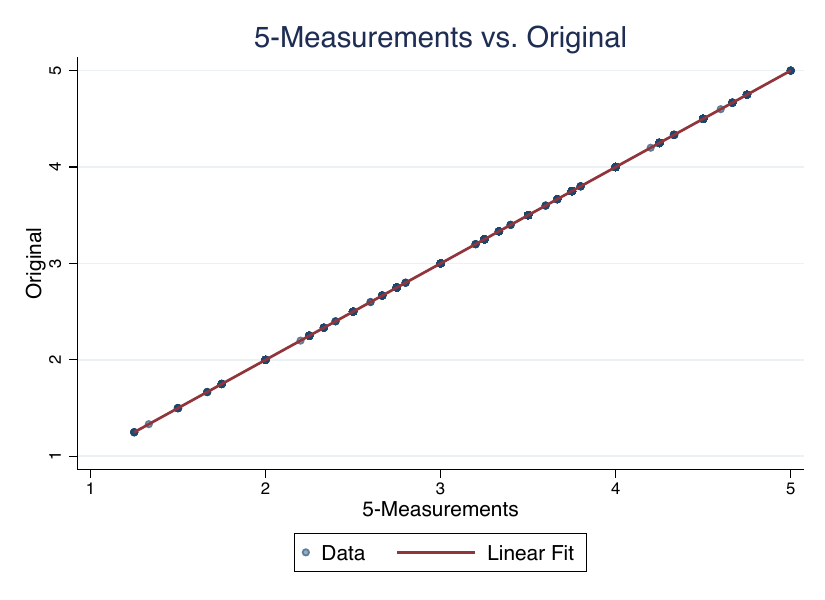}
    \begin{minipage}{1\linewidth}
    \footnotesize
    \textit{\textbf{Notes:}} This figure compares the original measure used in \cite{hussam_2022} with the one constructed by the author, for different number of measureements. 
    \end{minipage}
\end{figure}
\clearpage

\begin{figure}
    \centering
    \caption{Optimal Collection Plans by Budget Levels - Costly Measurements}
    \label{fig:optimal_n_t_c1}
    \includegraphics[width=0.44\linewidth]{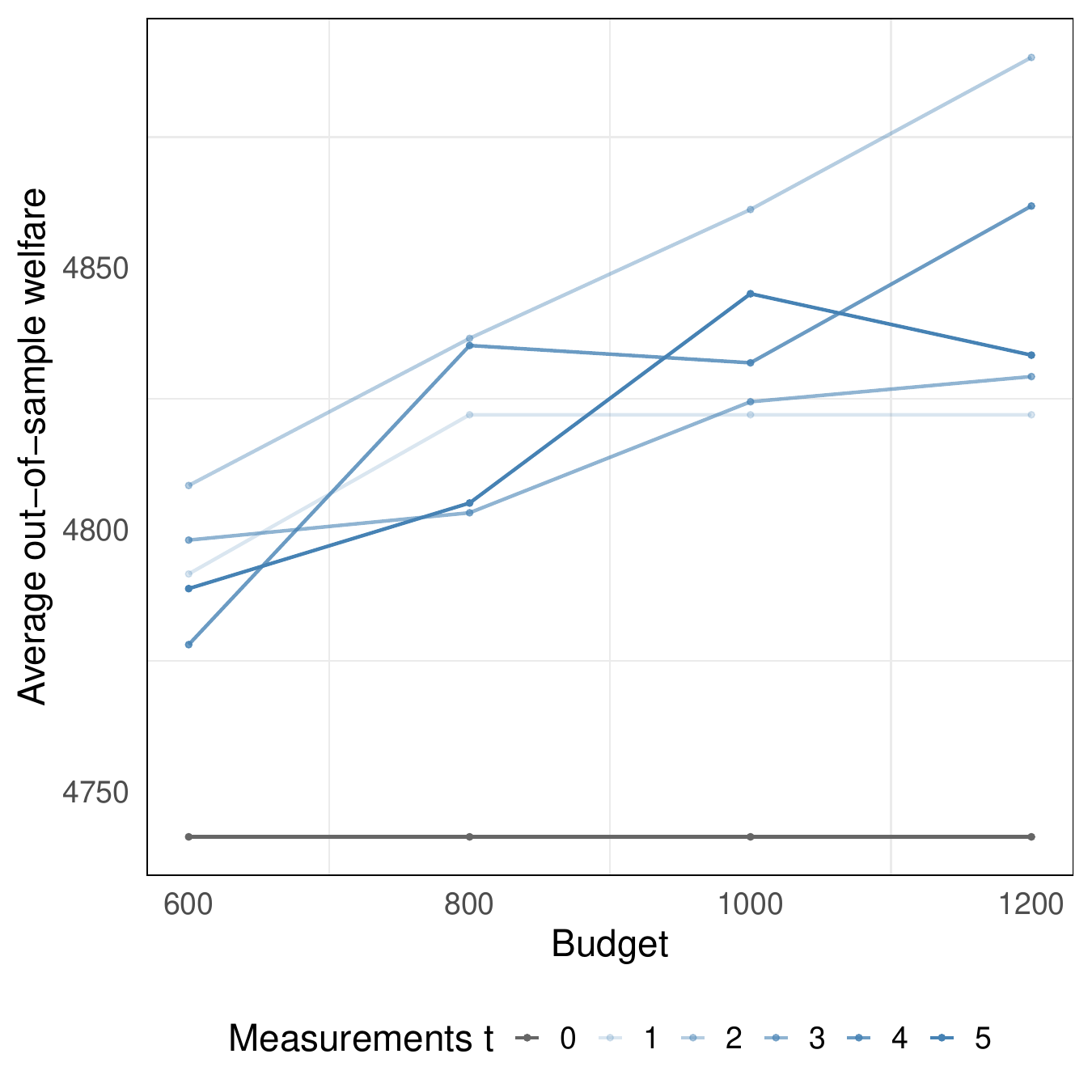}
    \includegraphics[width=0.44\linewidth]{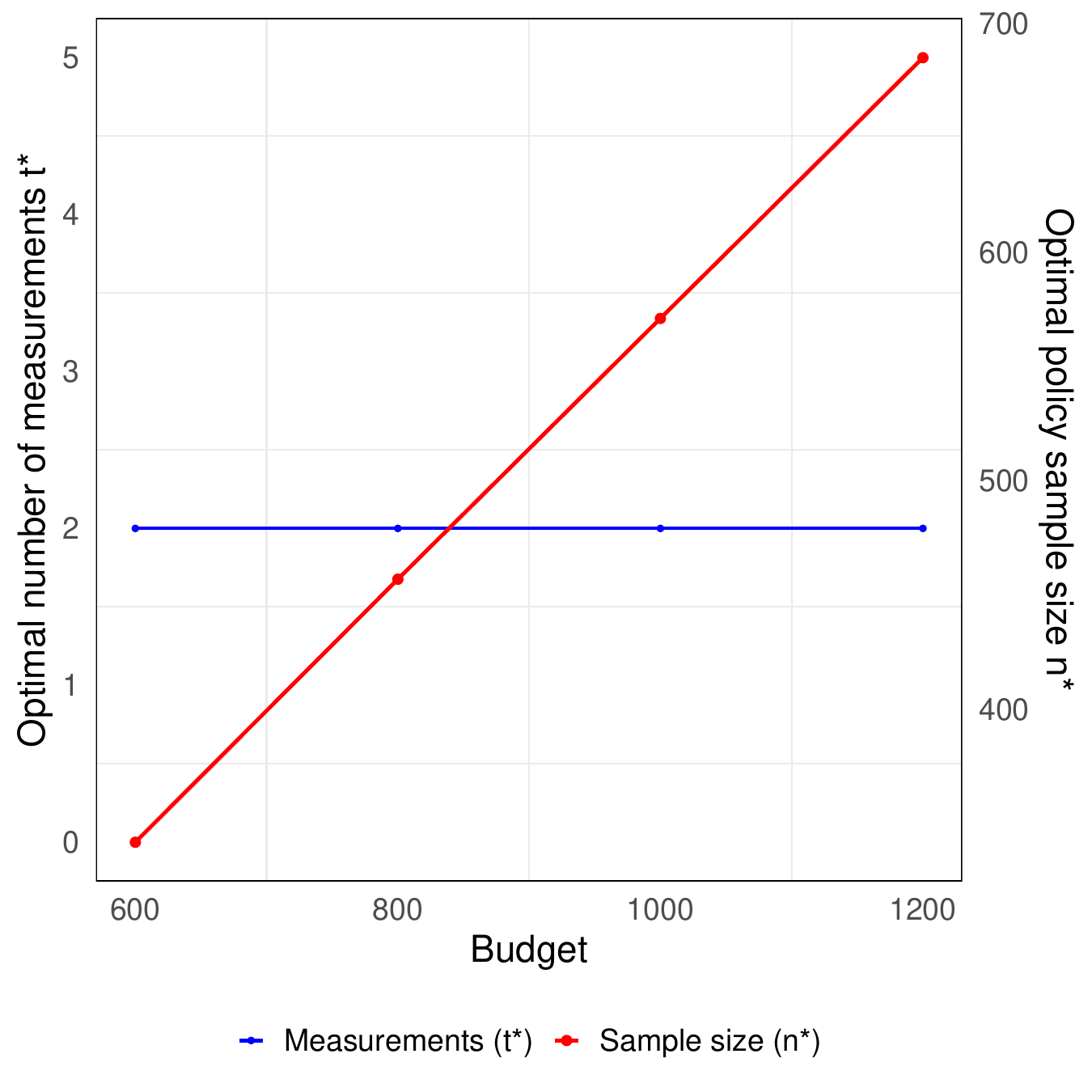}
    \begin{minipage}{1\linewidth}
    \footnotesize
    \textit{\textbf{Notes:}} This figure illustrates the performance of feasible data collection plans across different budget levels. The left panel reports the average out-of-sample welfare generated by designs with different numbers of measurements $t \in \{0,1,...,5\}$. 
    In this graph, we set $c_n = 0.25$ and $c_t=0.75$. 
    \end{minipage}
\end{figure}

\begin{figure}
    \centering
    \caption{Optimal Collection Plans by Budget Levels - Equal Marginal Cost}
    \label{fig:optimal_n_t_c2}
    \includegraphics[width=0.44\linewidth]{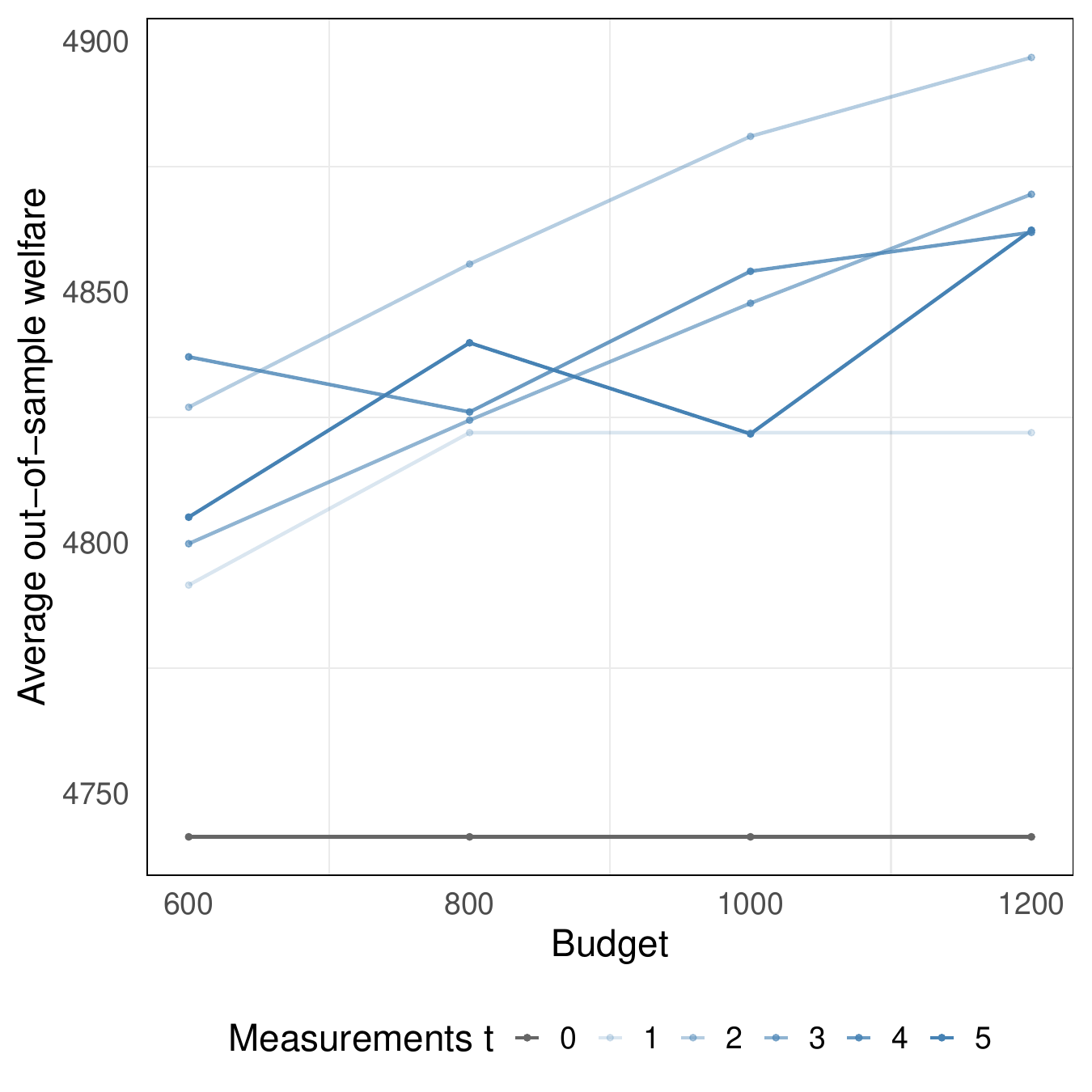}
    \includegraphics[width=0.44\linewidth]{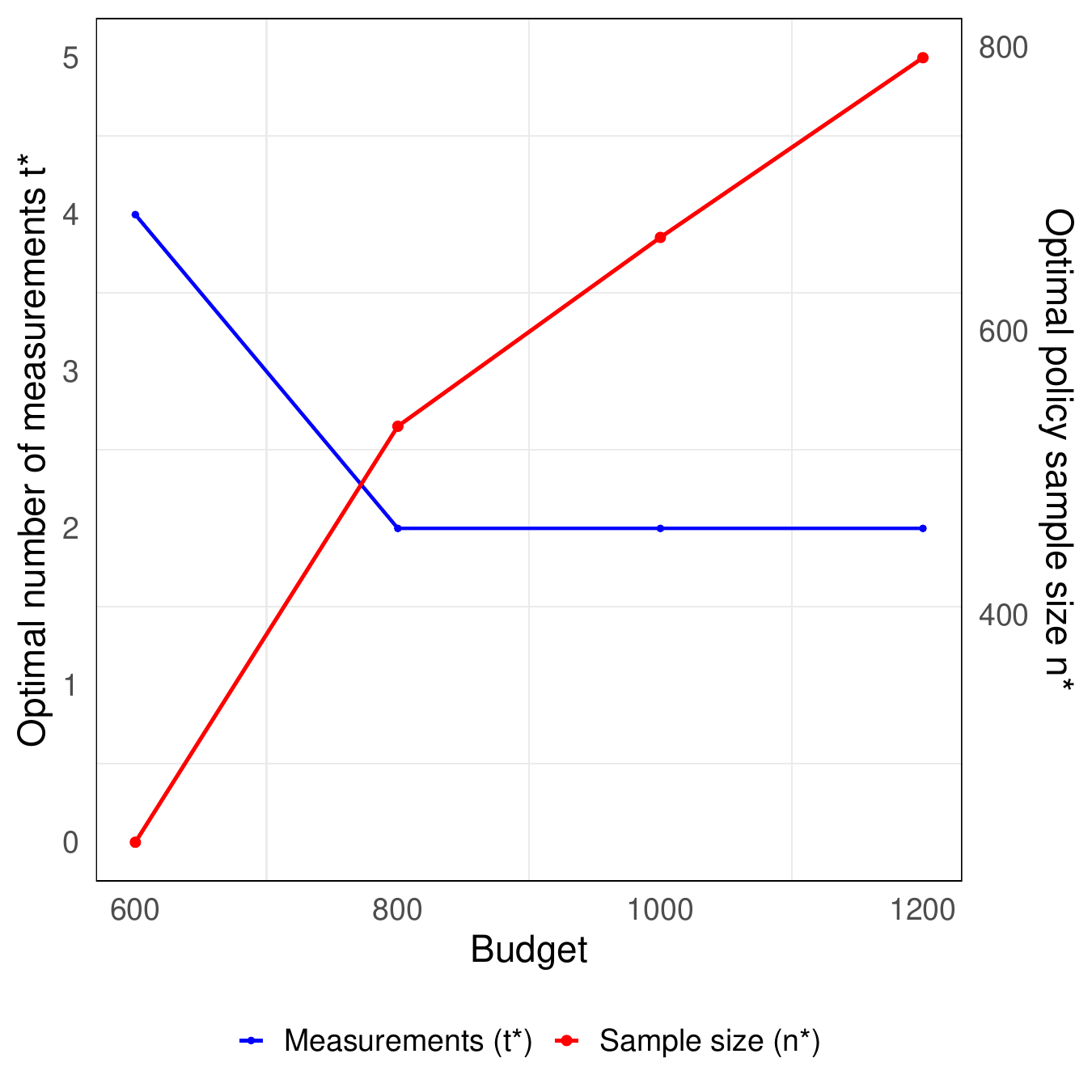}
    \begin{minipage}{1\linewidth}
    \footnotesize
    \textit{\textbf{Notes:}} This figure illustrates the performance of feasible data collection plans across different budget levels. The left panel reports the average out-of-sample welfare generated by designs with different numbers of measurements $t \in \{0,1,...,5\}$. 
    In this graph, we set $c_n = 0.5$ and $c_t=0.5$. 
    \end{minipage}
\end{figure}
\clearpage

\section{Formal Proofs}\label{app:formal_proofs}
\setcounter{equation}{0}
\renewcommand{\theequation}{B\arabic{equation}}

\begin{proof}[\textbf{Proof of Theorem \ref{thm:cb_upper}}]
Fix one $P \in \mathcal{P}$. Rewrite regret for Covariate-Based rules as: 
\begin{align}
    R(\hat{G}_{x}) 
    &= \mathbb{E}_{P^n}[W(G_{\theta}^*(X_i,A_i)) - W(\hat{G}_\theta(X_i))] \\
    &= \underbrace{W(G_{\theta}^*(X_i,A_i)) - W(G_{\theta}^*(X_i))}_{A} 
    + \underbrace{\mathbb{E}_{P^n}[W(G_{\theta}^*(X_i))- W(\hat{G}_\theta(X_i))]}_{B}
\end{align}

\textbf{Bounding component A.} Rewrite component $A$ as:
\begin{align}
    A & = W(G_\theta^*(X_i,A_i)) - W(G_\theta^*(X_i)) \\
    &\ \text{by definition of $W(G_z)$ and $G_z$:} \\
    & = \mathbb{E}_P[\tau_i \cdot G_{\theta}^*(X_i,A_i)]-  \mathbb{E}_P[\tau_i \cdot G_{\theta}^*(X_i)] \\
    & = \mathbb{E}_P[\tau_i \cdot (G_{\theta}^*(X_i,A_i)-G_{\theta}^*(X_i))] \\
    &\ \text{adding and subtracting $\mathbb{E}_P[\tau_i \cdot\mathbf{1}\{\tau(X_i)\geq 0\}]$:} \\
    & = \mathbb{E}_P[\tau_i \cdot (G_{\theta}^*(X_i,A_i) - \mathbf{1}\{\tau(X_i)\geq 0\})] + \mathbb{E}_P[\tau_i \cdot (\mathbf{1}\{\tau(X_i)\geq 0\} - G_{\theta}^*(X_i))] \\
    &\ \text{since}\ W(G^{FB}(X_i,A_i))\geq W(G^*_\theta(X_i,A_i)) \\
    &\ \text{implies}\ \mathbb{E}_P[\tau_i G^*_\theta(X_i,A_i)] \leq \mathbb{E}_P[\tau_i \mathbf{1}\{\tau(X_i,A_i)\geq 0\}], \\
    &\ \text{and $G^*_\theta(X_i)=\mathbf{1}\{s_{\theta^*}(X_i)\geq 0\}$:} \\
    & \leq \mathbb{E}_P[\tau_i \cdot (\mathbf{1}\{\tau(X_i,A_i)\geq 0\}-\mathbf{1}\{\tau(X_i)\geq 0\})] +  \\
    & \quad + \mathbb{E}_P[\tau_i \cdot (\mathbf{1}\{\tau(X_i)\geq 0\}-\mathbf{1}\{s_{\theta^*}(X_i)\geq 0\})] \\
    & := \mathbb{E}_{X,A}[\mathbb{E}_{P|X,A}[\tau_i \cdot (\mathbf{1}\{\tau(X_i,A_i)\geq 0\}-\mathbf{1}\{\tau(X_i)\geq 0\})]] + \Delta(s,\Theta_x) \\
    & = \mathbb{E}_{X,A}[\tau(X_i,A_i)\cdot (\mathbf{1}\{\tau(X_i,A_i)\geq 0\}-\mathbf{1}\{\tau(X_i)\geq 0\})] + \Delta(s,\Theta_x) \\
    &\ \text{because}\ \forall\ u,v \in \mathbb{R}, u\cdot(\mathbf{1}\{u\geq 0\}-\mathbf{1}\{v\geq 0\}) \leq |u-v|, \\ 
    & \leq \mathbb{E}_{X,A}[|\tau(X_i,A_i) - \tau(X_i)|] + \Delta(s,\Theta_x) \\
    &\ \text{by the law of iterated expectations and Jensen's inequality applied conditionally on $X_i$:} \\
    & \leq \mathbb{E}_X\left[\sqrt{\mathbb{E}_{A} [(\tau(X_i,A_i) - \tau(X_i))^2\ |\ X_i]}\right] + \Delta(s,\Theta_x)
    := \bar{\sigma}_{\tau|x} + \Delta(s,\Theta_x)
\end{align}

\textbf{Bounding component B.}
Component $B$ captures the welfare loss arising from maximizing the sample analog of the population welfare. Under Assumptions \ref{ass:kt18} and \ref{ass:score_funct}, Theorem 2.1 in \cite{kitagawa_who_2018} applies directly:
\begin{equation}
    \mathbb{E}_{P^n}[W(G_\theta^*(X_i))- W(\hat{G}_\theta(X_i))] 
    \leq C_1\frac{M}{k} \sqrt{\frac{v_x^\theta}{n}}
\end{equation}
where $C_1$ is a universal constant.

Therefore,
\begin{equation}
\mathbb{E}_{P^n}[W(G_\theta^*(X_i,A_i)) - W(\hat{G}_\theta(X_i))] 
\leq C_1\frac{M}{k} \sqrt{\frac{v_x^\theta}{n}} 
+ \bar{\sigma}_{\tau|x} + \Delta(s,\Theta_x)
\end{equation}

Finally, by definition of $\mathcal{P}(\sigma_0)$,
\begin{equation}
\sup_{P \in \mathcal{P}(\sigma_0)}\mathbb{E}_{P^n}[W(G_\theta^*(X_i,A_i)) - W(\hat{G}_\theta(X_i))] 
\leq C_1\frac{M}{k} \sqrt{\frac{v_x^\theta}{n}} 
+ \sigma_0 + \Delta(s,\Theta_x)
\end{equation}

Now, notice that:
\begin{align}
    \Delta(s,\Theta_x) & := \mathbb{E}_P[\tau_i \cdot (\mathbf{1}\{\tau(X_i)\geq 0\}-\mathbf{1}\{s_{\theta^*}(X_i)\geq 0\})] \\
    & \leq \mathbb{E}_{X}[\mathbb{E}_{P|X}[\tau_i \cdot (\mathbf{1}\{\tau(X_i)\geq 0\}-\mathbf{1}\{s_{\theta^*}(X_i)\geq 0\})|X_i]] \\
    & = \mathbb{E}_X[\tau(X_i)\cdot (\mathbf{1}\{\tau(X_i)\geq 0\}-\mathbf{1}\{s_{\theta^*}(X_i)\geq 0\})] \\
    & \leq \mathbb{E}_X[|\tau(X_i)|\cdot |\mathbf{1}\{\tau(X_i)\geq 0\}-\mathbf{1}\{s_{\theta^*}(X_i)\geq 0\}|] \\ 
    & \leq M \cdot \mathbb{P}_X(\mathbf{1}\{\tau(X_i)\geq 0\}\neq \mathbf{1}\{s_{\theta^*}(X_i)\geq 0\})
\end{align}
Finally, under Assumption \ref{ass:score_funct}.2, $\mathbb{P}_X(\mathbf{1}\{\tau(X_i)\geq 0\}\neq \mathbf{1}\{s_{\theta^*}(X_i)\geq 0\})=0$ because $\theta^*=\tilde{\theta}$ by the definition of first best. Therefore, if $\mathcal{G}_x^\theta$ satisfies Assumption \ref{ass:score_funct}.2, $\Delta(s,\Theta_x)=0$.

\end{proof} 

\begin{proof}[\textbf{Proof of Theorem \ref{thm:cb_lower}}]
Define the minimax risk \begin{equation} \mathcal{R}_n := \inf_{\{\hat{G}_\theta(X_i)\}} \sup_{P\in\mathcal{P}(\sigma_0)} R(\hat{G}_\theta(X_i)). \end{equation} We establish two separate lower bounds on $\mathcal{R}_n$: one arising from approximation error and one from estimation error, and then combine them. 

\textbf{Step 1: Approximation-error lower bound.} Fix $\sigma_0>0$ and consider the following data-generating process $P_\sigma$. Let $X_i\sim \mathcal{N}(\mu_x, \sigma_x^2)$, i.i.d.\ across $i$. 
Let $A_i \perp X_i$, and  $A_i \in \{-\sigma_0,+\sigma_0\}$ with probability $1/2$ each. Define 
\begin{equation} 
Y_i(0):=0, \qquad Y_i(1):=\tau(X_i,A_i):=A_i
\end{equation} 
Then, $Y_i(d)\in[-\sigma_0,\sigma_0]$, so Assumption \ref{ass:kt18}.1 holds for $\sigma_0\leq M/2$. Let $D_i\sim\mathrm{Bernoulli}(p)$ with $p\in(k,1-k)$, independent of $(X_i,A_i,Y_i(0),Y_i(1))$, so Assumptions \ref{ass:kt18}.2 and \ref{ass:kt18}.3 hold. 
Because $A_i$ is independent of $X_i$ and has mean zero, \begin{equation} 
m(X_i) := \mathbb{E}_P[\tau(X_i,A_i)\mid X_i] = \mathbb{E}_P[A_i\mid X_i] = 0 
\end{equation} 
Moreover, 
\begin{equation} 
\bar{\sigma}_{\tau|x}(P_\sigma) = \mathbb{E}_X \left[ \sqrt{\mathbb{V}_A(\tau(X_i,A_i)\mid X_i)} \right] = \sqrt{\mathbb{V}(A_i)} = \sigma_0
\end{equation} 
Hence $P_\sigma\in\mathcal{P}(\sigma_0)$. The oracle rule that observes $(X_i,A_i)$ is 
\begin{equation} 
G^*(X_i,A_i) = \mathbf{1}\{A_i>0\}
\end{equation} 
Its welfare is 
\begin{equation} 
W(G^*(X_i,A_i)) = \mathbb{E}[A_i \mathbf{1}\{A_i>0\}] = \frac{1}{2}\sigma_0 
\end{equation} 

For any $G_\theta (X_i)\in\mathcal{G}^\theta_x$, where $\mathcal{G}^\theta_x$ satisfies Assumption \ref{ass:score_funct},
\begin{equation} 
W(G_\theta(X_i)) = \mathbb{E}[A_i G_\theta(X_i)]
\end{equation} 
Since $\mathbb{E}[A_i\mid X_i]=0$, 
\begin{equation} 
W(G_\theta (X_i))=0 \quad \text{for all } G^\theta (X_i )\in\mathcal{G}^\theta_x 
\end{equation} 
Thus, $W(G^*_\theta(X_i))=0$ and, for any CB rule learned from data in the set $\{\hat{G}_\theta (X_i) \}$, 
\begin{equation} \mathbb{E}_{P^n}[W(\hat G_\theta(X_i))]=0
\end{equation} 
Therefore, 
\begin{equation} 
R_{P_\sigma}(\hat{G}_\theta (X_i)) = W(G_\theta^*(X_i,A_i)) = \frac{1}{2}\sigma_0 
\end{equation} 
Taking the infimum over $\{\hat{G}_\theta (X_i) \}$, 
\begin{equation}\label{eq:approx_lb_clean} 
    \mathcal{R}_n \ge \frac{1}{2}\sigma_0
\end{equation} 

\textbf{Step 2: Estimation-error lower bound.} We now invoke Theorem 2.2 of \citet{kitagawa_who_2018}. They construct a finite subclass $\mathcal{P}^*\subset\mathcal{P}$ with bounded outcomes, overlap $p\in(k,1-k)$, and covariates taking values in a set shattered by $\mathcal{G}^\theta_x$, such that for any sequence $\{\hat{G}_\theta (X_i) \}$, 
\begin{equation}\label{eq:kt_lower_clean} 
\sup_{P\in \mathcal{P}^*} \mathbb{E}_{P^n} \big[ W(G_\theta^*(X_i))-W(\hat{G}_\theta (X_i)) \big] \ge C_K\frac{M}{k}\sqrt{\frac{v_x^\theta}{n}} 
\end{equation} 
for some universal constant $C_K>0$. On $\mathcal{P}^*$, the treatment effect depends only on $X_i$, so $\bar{\sigma}_{\tau|x}(P)=0\le\sigma_0$, and hence $\mathcal{P}^*\subset\mathcal{P}(\sigma_0)$. 
Moreover, $G_\theta^*(X_i,A_i)=G^*_\theta (X_i)$ on $\mathcal{P}^*$, so 
\begin{equation} 
R(\hat{G}_\theta (X_i)) = \mathbb{E}_{P^n} \big[ W(G_\theta^*(X_i))-W(\hat{G}_\theta (X_i)) \big]
\end{equation} 
Thus, \eqref{eq:kt_lower_clean} implies 
\begin{equation}\label{eq:est_lb_clean} 
\mathcal{R}_n \ge C_K\frac{M}{k}\sqrt{\frac{v_x^\theta}{n}}
\end{equation} 

\textbf{Step 3: Combine the two bounds.}
From \eqref{eq:approx_lb_clean}, we have exhibited a single DGP $P_\sigma\in\mathcal P(\sigma_0)$ such that
\begin{equation}
\inf_{\{\hat{G}_\theta (X_i)\}}
\mathbb E_{(P_\sigma)^n}\!\left[ W(G^*_{\theta}(X_i,A_i)) - W(\hat{G}_\theta (X_i)) \right]
\ \ge\ \frac{1}{2}\sigma_0
\end{equation}
Therefore,
\begin{equation}\label{eq:cb_lb1_max}
\mathcal R_n
=
\inf_{\{\hat{G}_\theta (X_i)\}}
\sup_{P\in\mathcal P(\sigma_0)}
\mathbb E_{P^n}\!\left[ W(G^*_{\theta}(X_i,A_i)) - W(\hat{G}_\theta (X_i)) \right]
\ \ge\ \frac{1}{2}\sigma_0
\end{equation}
Similarly, \eqref{eq:est_lb_clean} implies
\begin{equation}\label{eq:cb_lb2_max}
\mathcal R_n
\ \ge\ C_K \frac{M}{k}\sqrt{\frac{v_x^\theta}{n}}
\end{equation}
Combining \eqref{eq:cb_lb1_max} and \eqref{eq:cb_lb2_max}, we conclude
\begin{equation}\label{eq:cb_max_combined}
\mathcal R_n
\ \ge\
\max\left\{
\frac{1}{2}\sigma_0, 
C_K \frac{M}{k}\sqrt{\frac{v_x^\theta}{n}}
\right\}
\end{equation}
Using the elementary inequality $\max\{u,v\}\ge (u+v)/2$ for all $u,v\ge 0$, we obtain
\begin{align}
\mathcal R_n
&\ge
\frac12\left(
\frac{1}{2}\sigma_0
+
C_K \frac{M}{k}\sqrt{\frac{v_x^\theta}{n}}
\right) \\
&=
\frac{1}{4}\sigma_0
+
\frac{C_K}{2}\frac{M}{k}\sqrt{\frac{v_x^\theta}{n}}
\end{align}
Setting $C_2:=C_K/2$ and $C_3:=1/4$ yields \eqref{eq:cb_minimax_lower}.
\end{proof}

\begin{proof}[\textbf{Proof of Theorem \ref{thm:ha_upper}}]
    Fix one $P \in \mathcal{P}$. Decompose regret into:
    \begin{equation}
        R(\hat{G}^{\theta}_{x,\hat{a}}) = \underbrace{W(G^*_{\theta}(X_i,A_i)) - W(G^*_{\theta}(X_i, \hat{A}_i))}_{I} + \underbrace{\mathbb{E}_{P^n}[W(G^*_{\theta}(X_i, \hat{A}_i)) - W(\hat{G}_{\theta}(X_i,\hat{A}_i))]}_{II}
    \end{equation}
    
    \textit{\textbf{Bounding Term I.}} Rewrite term I as:
    \begin{align}
        I & = W(G^*_{\theta}(X_i,A_i)) - W(G^*_{\theta}(X_i, \hat{A}_i)) \\
        & \text{by definition of $W(G^\theta_z)$ and $G^\theta_z$:} \\
        & = \mathbb{E}_P[\tau_i G_\theta^*(X_i,A_i)] - \mathbb{E}_{P}[\tau_iG_{\theta}^*(X_i,\hat{A}_i)] \\
        & \text{by definition of $G_\theta^*(X_i,A_i)$:} \\
        & = \sup_{\theta\in\Theta}\{\mathbb{E}_P[\tau_i G_{\theta}(X_i,A_i)]\} - \sup_{\theta \in \Theta}\{\mathbb{E}_P[\tau_i G_{\theta}(X_i,\hat{A}_i)]\} \\
        & \leq \sup_{\theta\in\Theta}\{\mathbb{E}_P[\tau_i \cdot(G_{\theta}(X_i,A_i) - G_{\theta}(X_i,\hat{A}_i))]\} \\
        & \leq \sup_{\theta\in\Theta}\{\mathbb{E}_P[|\tau_i|\cdot \mathbf{1}\{G_{\theta}(X_i,A_i) \neq G_{\theta}(X_i,\hat{A}_i)\}]\} \\
        & \leq M\cdot\sup_{\theta\in\Theta} \{ \mathbb{P}_P(G_\theta(X_i,A_i) \neq G_{\theta}(X_i,\hat{A}_i))\}
    \end{align}

    Now, notice that, if the two indicators disagree, the sign of $s_\theta(\cdot)$ must flip between $A_i$ and $\hat{A}_i$, which requires $|s_\theta(X_i,A_i)|$ to be no larger than the change $|s_\theta(X_i,A_i)-s_\theta(X_i,\hat{A}_i)|$. 
    Formally,
    \begin{equation}
        \{G_{\theta}(X_i,A_i) \neq G_{\theta}(X_i,\hat{A}_i)\} \subseteq \{|s_{\theta}(X_i,A_i)| \leq |s_{\theta}(X_i,A_i)-s_{\theta}(X_i,\hat{A}_i)|\}
    \end{equation}
    And, by Assumption \ref{ass:score_funct}.4 (Lipschitz score function):
    \begin{equation}
        \sup_{\theta \in \Theta} |s_{\theta}(X_i,A_i)-s_{\theta}(X_i,\hat{A}_i)| \leq L_s |\varepsilon_i|
    \end{equation}
    Therefore,
    \begin{align}
        \sup_{\theta\in\Theta} \{ \mathbb{P}_P(G_\theta(X_i,A_i) \neq G_{\theta}(X_i,\hat{A}_i))\} & \leq \sup_{\theta \in \Theta}\mathbb{P}_P(|s_\theta(X_i,A_i)|\leq  L_s |\varepsilon_i|) \\ 
        & = \sup_{\theta \in \Theta}\mathbb{E}_{X,\varepsilon}[\mathbb{P}_{P|X,\varepsilon}(|s_\theta(X_i,A_i)| \leq L_s |\varepsilon_i|\ |\ X_i,\varepsilon_i)] \\
        & \text{by Assumption \ref{ass:proxy}, } \varepsilon_i \perp A_i |X_i: \\
        & = \sup_{\theta \in \Theta}\mathbb{E}_{X,\varepsilon}[\mathbb{P}_{P|X}(|s_\theta(X_i,A_i)| \leq L_s |\varepsilon_i|\ |\ X_i)] \\
        & \leq \mathbb{E}_{X,\varepsilon}[\sup_{\theta \in \Theta}\mathbb{P}_{P|X}(|s_\theta(X_i,A_i)| \leq L_s |\varepsilon_i|\ |\ X_i)] \\
        &  \text{ by Assumption \ref{ass:score_funct}.3}: \\
        & \leq \mathbb{E}_X[\mathbb{E}_{\varepsilon}[\kappa L_s|\varepsilon_i|\ |\ X_i]] \\
        & \leq \kappa L_s\mathbb{E}_X[\mathbb{E}_{\varepsilon}[|\varepsilon_i|\ |\ X_i]] \\
        & \ \text{by the LIE and Jensen's applied conditionally on $X_i$.}: \\
        & \leq \kappa L_s\sqrt{\mathbb{E}_X\left[\mathbb{E}_{\varepsilon}[\varepsilon_i^2\ |\ X_i]\right]} \\
        & \leq \kappa L_s \sqrt{\bar{\sigma}_{\varepsilon}^2 + \bar{b}^2}
    \end{align}
    where $\bar{\sigma}^2_{\varepsilon}:=\mathbb{E}_X[\mathbb{V}(\varepsilon_i|X_i)]$ and $\bar{b}^2:= \mathbb{E}_X[\mathbb{E}[\varepsilon_i|X_i]^2]$. 
    
    Therefore, we can conclude:
    \begin{equation}
        I = W(G^*_{\theta}(X_i,A_i)) - W(G^*_{\theta}(X_i, \hat{A}_i)) \leq M \kappa L_s \sqrt{\bar{\sigma}_{\varepsilon}^2 + \bar{b}^2}
    \end{equation}
    
    \textit{\textbf{Bounding term II.}} Under Assumptions \ref{ass:kt18} and \ref{ass:score_funct}, the conditions for Theorem 2.1 in Kitagawa and Tetenov (2018) hold since treatment is randomized within $X_i$ and the propensity score $e(X_i)$ is known by Ass. \ref{ass:kt18}.2. Therefore, 
    \begin{equation}
        \mathbb{E}_{P^n}[W(G^*_{\theta}(X_i, \hat{A}_i))- W(\hat{G}_{\theta}(X_i,\hat{A}_i))] \leq C \frac{M}{k}\sqrt{\frac{v_{x,\hat{a}}^\theta}{n}}
    \end{equation}
    
    \textbf{Final bound.} Combining the upper bounds on component $I$ and $II$,
    \begin{align}
        \mathbb{E}_{P^n}[W(G^*_{\theta}(X_i,A_i))- W(\hat{G}_{\theta}(X_i,\hat{A}_i))] & \leq C \frac{M}{k}\sqrt{\frac{v_{x,\hat{a}}^\theta}{n}} +  M \kappa L_s\sqrt{\bar{\sigma}_{\varepsilon}^2 + \bar{b}^2} \\
        & = C \frac{M}{k}\sqrt{\frac{v_{x,\hat{a}}^\theta}{n}} +  M \kappa L_s\operatorname{rMSE}(\hat{A}_i)
    \end{align}
    By definition of $\mathcal{P}(\rho)$,
    \begin{equation}
        \sup_{P \in \mathcal{P}(\rho)}\mathbb{E}_{P^n}[W(G^*_{\theta}(X_i,A_i))- W(\hat{G}_{\theta}(X_i,\hat{A}_i))] \leq C \frac{M}{k}\sqrt{\frac{v_{x,\hat{a}}^\theta}{n}} +  M \kappa L_s\rho.
    \end{equation}
\end{proof}

\begin{proof}[\textbf{Proof of Theorem \ref{thm:ha_lower}}]
Define the minimax risk
\begin{equation}
\mathcal R_n^{\hat a}
:=
\inf_{\{\hat{G}_{\theta}(X_i,\hat{A}_i)\}}
\sup_{P\in\mathcal P(\rho)}
\mathbb E_{P^n}\!\left[W(G^*_{\theta}(X_i,A_i))-W(\hat{G}_{\theta}(X_i,\hat{A}_i))\right]
\end{equation}
We establish two separate lower bounds on $\mathcal R_n^{\hat a}$: one due to proxy information loss, and one due to estimating policies in a finite sample.

\textbf{Proxy Information Loss.}
Fix $\rho>0$ and consider the following DGP $P_\rho$.

Let $X_i\sim_{\text{i.i.d.}} \mathcal{N}(\mu_x,\sigma_x^2)$.
Let $A_i\sim\mathrm{Unif}[-1/\kappa, 1/\kappa]$ be independent of $X_i$, so that for all
$t\in[0,1/\kappa]$,
\begin{equation}
\mathbb P(|A_i|<t\mid X_i)=\mathbb P(|A_i|<t)=\kappa t .
\end{equation}
Let $\varepsilon_i\in\{-\rho,+\rho\}$ with probability $1/2$ each, independent of $(X_i,A_i)$, and define
$\hat A_i:=A_i+\varepsilon_i$. Then $\varepsilon_i\perp A_i\mid X_i$ and
$\mathbb E[\varepsilon_i^2]=\rho^2$, so Assumption \ref{ass:proxy} holds.

Define bounded potential outcomes by
\begin{equation}
Y_i(0):=-\frac{M}{2} \mathrm{sign}(A_i),
\qquad
Y_i(1):=+\frac{M}{2} \mathrm{sign}(A_i)
\end{equation}
with any convention $\mathrm{sign}(0)\in\{-1,+1\}$. Then $|Y_i(d)|\le M/2$, hence Assumption \ref{ass:kt18}.1 holds,
and $\tau_i:=Y_i(1)-Y_i(0)=M \mathrm{sign}(A_i)$.
Let $D_i\sim\mathrm{Bernoulli}(p)$ with $p\in(k,1-k)$ independent of $(X_i,A_i,Y_i(0),Y_i(1))$, so
Assumptions \ref{ass:kt18}.2--3 hold. Therefore $P_\rho\in\mathcal P(\rho)$.

The oracle that observes $A_i$ treats iff $A_i\ge 0$, i.e.\ uses $G_\theta(x,a)=\mathbf 1\{a\ge 0\}\in\mathcal G^\theta_{x,a}$. Note that this class satisfies Assumption \ref{ass:score_funct} since it has finite VC dimension (Ass. \ref{ass:score_funct}.1), it satisfies the margin condition (Ass. \ref{ass:score_funct}.3) with constant $\kappa$, and it is Lipschitz continuous (Ass. \ref{ass:score_funct}.4) with constant $L_s=1$.
Under this rule, the realized outcome equals $+M/2$, hence
\begin{equation}
W(G^*_{\theta}(X_i,A_i))=\frac{M}{2}
\end{equation}

For any policy $G(Z_i)$,
\begin{align}
Y_i(1)G(Z_i)+Y_i(0)(1-G(Z_i)) & = \frac{M}{2}\operatorname{sign}(A_i)G(Z_i) - \frac{M}{2}\operatorname{sign}(A_i)(1-G(Z_i))\\
& = \frac{M}{2}\mathrm{sign}(A_i) (2G(Z_i)-1)
\end{align}
Therefore, for any policy $G_\theta(X_i,\hat{A}_i)$,
\begin{align}
W(G_\theta(X_i,\hat{A}_i))&=\mathbb E_P\left[\frac{M}{2} \mathrm{sign}(A_i)(2G_\theta(X_i,\hat{A}_i)-1)\right] \\
&\ \text{by developing the expectation:} \\
& = \frac{M}{2}\left( \mathbb{P}_P(A_i \geq 0) \cdot \left[ \mathbb{P}_P(G_{\theta}(X_i,\hat{A}_i) = \mathbf{1}\{A_i\geq 0\}) -  \mathbb{P}_P(G_{\theta}(X_i,\hat{A}_i) \neq \mathbf{1}\{A_i\geq 0\}) \right]  \right) - \\
& - \frac{M}{2}\left( \mathbb{P}_P(A_i < 0) \cdot \left[ \mathbb{P}_P(G_{\theta}(X_i,\hat{A}_i) \neq \mathbf{1}\{A_i\geq 0\}) -  \mathbb{P}_P(G_{\theta}(X_i,\hat{A}_i) = \mathbf{1}\{A_i\geq 0\}) \right]  \right) \\
& = \frac{M}{2}\left(\mathbb{P}_P(G_\theta(X_i,\hat{A}_i)= \mathbf{1}\{A_i\geq 0\}) - \mathbb{P}_P(G_\theta(X_i,\hat{A}_i) \neq \mathbf{1}\{A_i\geq 0\}) \right) \\
&\ \text{because the two events are complementary,} \\
& =\frac{M}{2}\left(1-2\mathbb P_P(G_\theta(X_i,\hat{A}_i)\neq \mathbf 1\{A_i\ge 0\})\right)
\end{align}

Therefore, the welfare-maximizing rule in $\mathcal G^\theta_{x,\hat a}$, denoted $G_\theta^* (X_i, \hat{A}_i)$,
is the Bayes classifier of the label $\mathbf 1\{A_i\ge 0\}$ given $(X_i,\hat A_i)$.

Consider the event $\mathcal E:=\{|\hat A_i|<\rho\}$. For any fixed $\hat a\in(-\rho,\rho)$, the two values of $A_i$
compatible with $\hat A_i=\hat a$ are $\hat a-\rho<0$ and $\hat a+\rho>0$. Since $\varepsilon_i$ is symmetric and
independent of $A_i$, it follows that
\begin{equation}
\mathbb P(A_i\ge 0\mid \hat A_i=\hat a)=\mathbb P(A_i<0\mid \hat A_i=\hat a)=\frac12
\qquad \forall \hat a\in(-\rho,\rho)
\end{equation}
so the Bayes conditional classification error equals $1/2$ on $\mathcal E$ and hence
\begin{equation}
\mathbb P\!\left(G_\theta^* (X_i, \hat{A}_i)\neq \mathbf 1\{A_i\ge 0\}\right)
\ \ge\ \frac12 \mathbb P(|\hat A_i|<\rho)
\end{equation}
Next compute $\mathbb P(|\hat A_i|<\rho)$. If $\varepsilon_i=-\rho$, then
$|\hat A_i|<\rho \iff A_i\in(0,2\rho)$; if $\varepsilon_i=+\rho$, then
$|\hat A_i|<\rho \iff A_i\in(-2\rho,0)$. Therefore,
\begin{equation}
\mathbb P(|\hat A_i|<\rho)
=
\frac12 \mathbb P(0<A_i<2\rho)+\frac12 \mathbb P(-2\rho<A_i<0)
\end{equation}
Because $A_i\sim\mathrm{Unif}[-1/\kappa,1/\kappa]$, for $\rho\le 1/(2\kappa)$,
\begin{equation}
\mathbb P(0<A_i<2\rho)=\mathbb P(-2\rho<A_i<0)=\frac{2\rho}{2/\kappa}=\kappa\rho
\end{equation}
hence $\mathbb P(|\hat A_i|<\rho)=\kappa\rho$.
Thus,
\begin{equation}
\mathbb P\!\left(G_\theta^* (X_i, \hat{A}_i)\neq \mathbf 1\{A_i\ge 0\}\right)
\ \ge\ \frac12 \kappa\rho
\end{equation}
Therefore,
\begin{equation}
W(G^*_{\theta}(X_i,A_i))-W(G_\theta^* (X_i, \hat{A}_i))
\ \ge\ \frac{M}{2}\kappa\rho
\end{equation}
Since $W(\hat{G}_{\theta}(X_i,\hat{A}_i))\le W(G_\theta^* (X_i, \hat{A}_i))$ for any estimator $\hat{G}_{\theta}(X_i,\hat{A}_i)$
taking values in $\mathcal G^\theta_{x,\hat a}$, we conclude
\begin{equation}\label{eq:info_lb_hat}
\inf_{\{\hat{G}_{\theta}(X_i,\hat{A}_i)\}}
\mathbb E_{(P_\rho)^n}\!\left[W(G^*_{\theta}(X_i,A_i))-W(\hat{G}_{\theta}(X_i,\hat{A}_i))\right]
\ \ge\ \frac{M}{2}\kappa\rho
\end{equation}
In particular, $\mathcal R_n^{\hat a}\ge \frac{M}{2}\kappa\rho$.

\textbf{Statistical error.}
Invoke Theorem 2.2 of \citet{kitagawa_who_2018}: there exists a finite subclass $\mathcal{P}^*$ satisfying
Assumption \ref{ass:kt18}.1--3 and such that the covariates take values in a set shattered by
$\mathcal G^\theta_{x,\hat a}$ (hence by VC-dimension $v^\theta_{x,\hat a}$), for which
\begin{equation}\label{eq:kt_lower_hat}
\sup_{P\in \mathcal{P}^*}\ \mathbb E_{P^n}\!\left[W(G_\theta^* (X_i, \hat{A}_i))-W(\hat{G}_{\theta}(X_i,\hat{A}_i))\right]
\ \ge\
C_K \frac{M}{k}\sqrt{\frac{v^\theta_{x,\hat a}}{n}}
\end{equation}
for a universal constant $C_K>0$ (with the dependence on $k$ as stated in KT18).

Choose the KT18 subclass $\mathcal{P}^*$ so that $\varepsilon_i\equiv 0$ almost surely (i.e.\ $\hat A_i=A_i$). Then
$\sqrt{\mathbb E[\varepsilon_i^2]}=0\le \rho$, so $\mathcal{P}^*\subset\mathcal P(\rho)$.
Moreover, on $\mathcal{P}^*$ we have $\hat A_i=A_i$, hence $\mathcal G^\theta_{x,\hat a}$ and $\mathcal G^\theta_{x,a}$ coincide
pointwise and therefore $G_\theta^* (X_i, \hat{A}_i)=G^*_{\theta}(X_i,A_i)$.
Thus \eqref{eq:kt_lower_hat} implies
\begin{equation}\label{eq:est_lb_hat}
\mathcal R_n^{\hat a}\ \ge\ C_K \frac{M}{k}\sqrt{\frac{v^\theta_{x,\hat a}}{n}}
\end{equation}

\textbf{Step 3: Combine the two bounds.}
From \eqref{eq:info_lb_hat}, we have exhibited a single DGP $P_\rho\in\mathcal P(\rho)$ such that
\begin{equation}
\inf_{\{\hat{G}_{\theta}(X_i,\hat{A}_i)\}}
\mathbb E_{(P_\rho)^n}\!\left[W(G^*_{\theta}(X_i,A_i))-W(\hat{G}_{\theta}(X_i,\hat{A}_i))\right]
\ \ge\ \frac{M}{2}\kappa\rho
\end{equation}
Therefore,
\begin{equation}\label{eq:lb1_max}
\mathcal R_n^{\hat a}
=
\inf_{\{\hat{G}_{\theta}(X_i,\hat{A}_i)\}}
\sup_{P\in\mathcal P(\rho)}
\mathbb E_{P^n}\!\left[W(G^*_{\theta}(X_i,A_i))-W(\hat{G}_{\theta}(X_i,\hat{A}_i))\right]
\ \ge\ \frac{M}{2}\kappa\rho
\end{equation}
Similarly, \eqref{eq:est_lb_hat} implies
\begin{equation}\label{eq:lb2_max}
\mathcal R_n^{\hat a}
\ \ge\ C_K \frac{M}{k}\sqrt{\frac{v^\theta_{x,\hat a}}{n}}
\end{equation}
Combining \eqref{eq:lb1_max} and \eqref{eq:lb2_max}, we conclude
\begin{equation}\label{eq:max_combined}
\mathcal R_n^{\hat a}
\ \ge\
\max\left\{
\frac{M}{2}\kappa\rho, 
C_K \frac{M}{k}\sqrt{\frac{v^\theta_{x,\hat a}}{n}}
\right\}
\end{equation}
Using the elementary inequality $\max\{u,v\}\ge (u+v)/2$ for all $u,v\ge 0$, we obtain
\begin{align}
\mathcal R_n^{\hat a}
&\ge
\frac12\left(
\frac{M}{2}\kappa\rho
+
C_K \frac{M}{k}\sqrt{\frac{v^\theta_{x,\hat a}}{n}}
\right) \\
&=
\frac{M}{4}\kappa\rho
+
\frac{C_K}{2}\frac{M}{k}\sqrt{\frac{v^\theta_{x,\hat a}}{n}}
\end{align}
Setting $C_4:=C_K/2$ and $C_5:=1/4$ yields \eqref{eq:ha_lower}.
\end{proof}

\begin{proof}[\textbf{Proof of Proposition \ref{prop:minimax_design}}]
By Theorem \ref{thm:cb_upper}, together with Assumption \ref{ass:prior_sigma}, the feasible Covariate-Based design satisfies
\begin{equation}
    V_{CB}(B_0)\le \overline{R}_{CB}(B_0).
\end{equation}
By Theorem \ref{thm:cb_lower}, together with the same extension argument used in Corollary \ref{cor:proxy_domination},
\begin{equation}
    \underline{R}_{CB}(B_0)\le V_{CB}(B_0).
\end{equation}
Similarly, Theorems \ref{thm:ha_upper} and \ref{thm:ha_lower} imply that, for every $t\in\mathcal{T}$,
\begin{equation}
    \underline{R}_{A}(t,B_0)\le V_A(t,B_0)\le \overline{R}_{A}(t,B_0).
\end{equation}

To prove part 1, suppose that, for every $t\in\mathcal{T}$,
\begin{equation}
    \overline{R}_{CB}(B_0) < \underline{R}_{A}(t,B_0).
\end{equation}
Then, for every $t\in\mathcal{T}$,
\begin{equation}
    V_{CB}(B_0)\le \overline{R}_{CB}(B_0) < \underline{R}_{A}(t,B_0)\le V_A(t,B_0),
\end{equation}
which implies
\begin{equation}
    V_{CB}(B_0) < V_A(t,B_0)\qquad \forall t\in\mathcal{T}.
\end{equation}

To prove part 2, suppose that
\begin{equation}
    \overline{R}_{A}(t^*,B_0)
    <
    \min \left\{\underline{R}_{CB}(B_0), \inf_{t\in\mathcal{T}\setminus\{t^*\}}\underline{R}_{A}(t,B_0)\right\}.
\end{equation}
Then, for every $t\neq t^*$,
\begin{equation}
    V_A(t^*,B_0)\le \overline{R}_{A}(t^*,B_0) < \underline{R}_{A}(t,B_0)\le V_A(t,B_0),
\end{equation}
which proves that $V_A(t^*,B_0) < V_A(t,B_0)$ for all $t\neq t^*$. Moreover,
\begin{equation}
    V_A(t^*,B_0)\le \overline{R}_{A}(t^*,B_0) < \underline{R}_{CB}(B_0)\le V_{CB}(B_0),
\end{equation}
Which prooves the dominance with respect to CB rules.
\end{proof}

\section{Additional Results}
\subsection{Examples' Proofs}\label{app:examples}
\begin{proof}[Proof of Example \ref{ex:repeated_measure}]\label{proof:ex:repeated_measure}
 We can write the measurement error as:
\begin{equation}
    \hat{A}_i(t)-A_i=\frac{1}{t}\sum_{j=1}^t U_{ij}.
\end{equation}

Define
\begin{equation}
    b_P(X_i):=\mathbb{E}_P[U_{ij}\mid X_i],
    \qquad
    \sigma^2_{U,P}(X_i):=\mathbb{V}_P(U_{ij}\mid X_i).
\end{equation}
Using conditional independence across $j$ given $X_i$,
\begin{align}
    \mathbb{E}_P\left[(\hat{A}_i(t)-A_i)^2\mid X_i\right]
    &=
    \mathbb{V}_P\left(\frac{1}{t}\sum_{j=1}^t U_{ij}\,\middle|\,X_i\right)
    +
    \left(
    \mathbb{E}_P\left[\frac{1}{t}\sum_{j=1}^t U_{ij}\,\middle|\,X_i\right]
    \right)^2 \\
    &=\frac{1}{t^2}\sum_{j=1}^t \mathbb{V}_P(U_{ij}\mid X_i) + b_P(X_i)^2 \\
    &=\frac{\sigma^2_{U,P}(X_i)}{t}+b_P(X_i)^2.
\end{align}
Taking expectations over $X_i$ yields
\begin{equation}
    \mathbb{E}_P\left[(\hat{A}_i(t)-A_i)^2\right]
    =\frac{1}{t}\,\mathbb{E}_P\left[\sigma^2_{U,P}(X_i)\right]
    +
    \mathbb{E}_P\left[b_P(X_i)^2\right].
\end{equation}
Therefore,
\begin{align}
    \mathrm{rMSE}(\hat{A}_i(t))
    &:=
    \sqrt{\mathbb{E}_P\left[(\hat{A}_i(t)-A_i)^2\right]} \\
    &=
    \sqrt{
    \frac{1}{t}\,\mathbb{E}_P\left[\sigma^2_{U,P}(X_i)\right]
    +
    \mathbb{E}_P\left[b_P(X_i)^2\right]
    } \\
    &\le
    \sqrt{\frac{1}{t}\,\mathbb{E}_P\left[\sigma^2_{U,P}(X_i)\right]}
    +
    \sqrt{\mathbb{E}_P\left[b_P(X_i)^2\right]} \\
    &\le
    \frac{m_0}{\sqrt{t}} + b_0.
\end{align}
\noindent where the first inequality follows from $\sqrt{u+v}\le \sqrt{u}+\sqrt{v}$ for $u,v\ge 0$, and the second from the uniform bounds in Example \ref{ex:repeated_measure}. Hence, in the repeated-measurement case, Assumption \ref{ass:time} is satisfied with the envelope
\begin{equation}
    h(t)=b_0+\frac{m_0}{\sqrt{t}}.
\end{equation}    
\end{proof}

\subsection{External Data-Dependent Proxy} \label{app:external_proxy}
\setcounter{assumption}{2}
\renewcommand{\theassumption}{2B}

\begin{assumption}[External data-dependent $\hat{A}_i$]\label{ass:ext_proxy} $ $
\begin{enumerate}
    \item \textbf{Estimate Representation} - Let $\hat{A}_i$ be written as $\hat A_i=\hat{f}_m(X_i)$.
    \item \textbf{External Estimator} - $\hat{f}_m:\mathcal X\to\hat{\mathcal A}$ is learned on an auxiliary sample
    $S_m := \{(Y_i(0),X_i)\}_{i=1}^m \perp S_n$ and then treated as fixed in the policy-learning sample $S_n$.
\end{enumerate}
\end{assumption}

\renewcommand{\theassumption}{\arabic{assumption}}
\setcounter{assumption}{3}
\renewcommand{\theassumption}{3B}

\begin{assumption}[Policy class restrictions]\label{ass:score_funct_ext} 
$ $
\begin{enumerate}
    \item \textbf{VC Class} - The policy class $\mathcal{G}^\theta_{z}$ has finite VC-dimension $v^\theta_z<\infty$.
    \item \textbf{Margin Condition} - There exists a constant $\kappa>0$ such that, for all $t\geq0$:
    \begin{equation}
        \sup_{\theta \in \Theta}\mathbb{P}(|s_\theta(X_i,A_i)|<t|X_i=x, S_m) \leq \kappa t \quad \forall\ x \in \mathcal{X}
    \end{equation} 
    \item \textbf{Lipschitz Continuity} - There exists a constant $L_s$ such that:
    \begin{equation}
        \sup_{\theta\in\Theta,\ (x,a) \in \mathcal{X}\times \mathcal{A}}|s_{\theta}(x,a)- s_{\theta}(x,a+\gamma)| \leq L_s |\gamma|
    \end{equation}
\end{enumerate}
\end{assumption}
\setcounter{assumption}{5}

\begin{proposition}[Regret bound for $\hat{a}$-Augmented rules when $\hat{a}$ is learned externally]\label{prop:ha_upper_ext}
Under Assumptions \ref{ass:kt18} and \ref{ass:ext_proxy}, the regret of any $\hat{a}$-CB policy class $\mathcal{G}^\theta_{x,\hat{a}}$ that satisfies Assumption \ref{ass:score_funct_ext} satisfies:
\begin{equation}
\sup_{P\in\mathcal P}
\mathbb E_{P^n}\!\left[W(G^*_{\theta}(X_i,A_i))-W(\hat{G}_{\theta}(X_i,\hat{A}_i))\right]
\le
C_1\frac{M}{k}\sqrt{\frac{v^\theta_{x,\hat a}}{n}}
+
M\kappa L_s \operatorname{rMSE}_m(\hat{A}_i)
\end{equation}
where $C_1$ is a universal constant, and
\begin{equation}
        \operatorname{rMSE}_m(\hat{A}_i) := \sqrt{{\mathbb{E}_P\left[ (\hat{A}_i-A_i)^2 | S_m \right]}}
\end{equation}
\end{proposition}
The formal proof is reported in Appendix \ref{app:formal_proofs_ext}.

\begin{proposition}[Minimax lower bound for $\hat{a}$-Augmented rules when $\hat{a}$ is learned externally]
\label{prop:ha_lower_ext}
Fix $\rho>0$. Let $\mathcal{P}(\rho)$ denote the class of data-generating processes $P$ satisfying Assumptions \ref{ass:kt18} and \ref{ass:ext_proxy}, and such that $\operatorname{rMSE}_m(\hat{A}_i)\le \rho\leq 1/(2\kappa)$.
Then, for any class $\mathcal G^\theta_{x,\hat{a}}$ that satisfies Assumption \ref{ass:score_funct_ext}, there exist universal constants $C_2>0$ and $C_3>0$ such that:
\begin{equation}\label{eq:ha_minimax_lower}
\inf_{\{\hat{G}_{\theta}(X_i,\hat{A}_i)\}}
\sup_{P\in \mathcal P(\rho)}
\mathbb E_{P^n}\!\left[W(G^*_{\theta}(X_i,A_i))-W(\hat{G}_{\theta}(X_i,\hat{A}_i))\right]
\ge
C_2 \frac{M}{k}\sqrt{\frac{v^\theta_{x,\hat a}}{n}}
+
C_4  M \kappa \rho
\end{equation}
\end{proposition}
The formal proof is reported in Appendix \ref{app:formal_proofs_ext}.

\subsubsection{Formal Proofs} \label{app:formal_proofs_ext}
\begin{proof}[\textbf{Proof of Proposition \ref{prop:ha_upper_ext}}]
Decompose regret as
\begin{equation}
R(\hat{G}_{\theta}(X_i,\hat{A}_i))
=
\underbrace{W(G^*_{\theta}(X_i,A_i))-W(G_\theta^* (X_i, \hat{A}_i))}_{I}
+
\underbrace{\mathbb E_{P^n}\!\left[W(G_\theta^* (X_i, \hat{A}_i))-W(\hat{G}_{\theta}(X_i,\hat{A}_i))\right]}_{II}
\end{equation}

\textit{\textbf{Bounding $I$.}} Rewrite term $I$ as:
\begin{align}
    I 
    &= W(G^*_{\theta}(X_i,A_i)) - W(G^*_{\theta}(X_i,\hat{A}_i)) \\
    & \text{by definition of $W(G^\theta_z)$ and $G^\theta_z$:} \\
    &= \mathbb{E}_P[\tau_i G_\theta^*(X_i,A_i)] - \mathbb{E}_{P}[\tau_iG_{\theta}^*(X_i,\hat{A}_i)] \\
    & \text{by definition of $G_\theta^*(X_i,A_i)$:} \\
    &= \sup_{\theta\in\Theta}\{\mathbb{E}_P[\tau_i G_{\theta}(X_i,A_i)]\} - \sup_{\theta \in \Theta}\{\mathbb{E}_P[\tau_i G_{\theta}(X_i,\hat{A}_i)]\} \\
    & \leq \sup_{\theta\in\Theta}\{\mathbb{E}_P[\tau_i \cdot(G_{\theta}(X_i,A_i) - G_{\theta}(X_i,\hat{A}_i))]\} \\
    & \leq \sup_{\theta\in\Theta}\{\mathbb{E}_P[|\tau_i|\cdot \mathbf{1}\{G_{\theta}(X_i,A_i) \neq G_{\theta}(X_i,\hat{A}_i)\}]\} \\
    & \leq M\cdot\sup_{\theta\in\Theta} \{ \mathbb{P}_P(G_\theta(X_i,A_i) \neq G_{\theta}(X_i,\hat{A}_i))\}
\end{align}

Now, notice that:
\begin{equation}
    \{G_{\theta}(X_i,A_i) \neq G_{\theta}(X_i,\hat{A}_i)\} \subseteq \{|s_{\theta}(X_i,A_i)| \leq |s_{\theta}(X_i,A_i)-s_{\theta}(X_i,\hat{A}_i)|\}
\end{equation}
And, by Assumption \ref{ass:score_funct_ext}.3 (Lipschitz score function):
\begin{equation}
    \sup_{\theta \in \Theta} |s_{\theta}(X_i,A_i)-s_{\theta}(X_i,\hat{A}_i)| \leq L_s |A_i-\hat{A}_i|
\end{equation}
Therefore,
\begin{align}
    \sup_{\theta\in\Theta} \{ \mathbb{P}_P(G_\theta(X_i,A_i) & \neq G_{\theta}(X_i,\hat{A}_i))\}
     \leq \sup_{\theta \in \Theta}\mathbb{P}_P(|s_\theta(X_i,A_i)|\leq  L_s |A_i-\hat{A}_i|) \\ 
    & = \sup_{\theta \in \Theta}\mathbb{E}_{X,A|S_m}\!\left[\mathbf{1}\{|s_\theta(X_i,A_i)| \leq L_s |A_i-\hat{A}_i|\}\mid S_m\right] \\
    & = \sup_{\theta \in \Theta}\mathbb{E}_{X,A|S_m}\!\left[\mathbb{P}_{P|X,S_m}\!\left(|s_\theta(X_i,A_i)| \leq L_s |A_i-\hat{A}_i|\mid X_i, S_m\right)\right] \\
    & \leq \sup_{\theta \in \Theta}\mathbb{E}_{X,A|S_m}\!\left[\mathbb{P}_{P|X,S_m}\!\left(|s_\theta(X_i,A_i)| \leq L_s |A_i-\hat{A}_i|\mid X_i, S_m\right)\right]\\
    & \text{by Assumption \ref{ass:score_funct_ext}.2:} \\
    & \leq \mathbb{E}_{X,A|S_m}\!\left[\kappa L_s |A_i-\hat{A}_i|\mid S_m \right]\\
    & = \kappa L_s\mathbb{E}_{X,A|S_m}[|A_i-\hat{A}_i| \mid S_m]\\
    & \text{by }|A_i-\hat{A}_i| = \sqrt{(A_i-\hat{A}_i)^2}\text{ and Jensen's inequality:} \\
    & \leq \kappa L_s \sqrt{\mathbb{E}_{X,A|S_m}[(A_i-\hat{A}_i)^2\mid S_m]} \\
    & = \kappa L_s\operatorname{rMSE}_m(\hat{A}_i)
\end{align}

Therefore, we can conclude:
\begin{equation}
    I = W(G^*_{\theta}(X_i,A_i)) - W(G^*_{\theta}(X_i,\hat{A}_i))
    \leq M \kappa L_s \operatorname{rMSE}_m(\hat{A}_i)
\end{equation}

\textit{\textbf{Bounding $II$.}}
Conditional on $S_m$ (hence on $\hat A_i=\hat f_m(X_i)$), the sample $\{(Y_i,X_i,\hat A_i,D_i)\}_{i=1}^n$ is i.i.d.
and Assumption \ref{ass:score_funct_ext}.1 holds with VC dimension $v_{x,\hat a}^\theta$.
Therefore, conditional on $S_m$, Theorem 2.1 of \citet{kitagawa_who_2018} implies
\begin{equation}
\mathbb E_{P^n}\!\left[W(G_\theta^* (X_i, \hat{A}_i))-W(\hat{G}_{\theta}(X_i,\hat{A}_i)) \mid S_m \right]
\le
C \frac{M}{k}\sqrt{\frac{v^\theta_{x,\hat a}}{n}}
\end{equation}
for a universal constant $C>0$. Since $S_n \perp S_m$,
\begin{equation}
\mathbb E_{P^n}\!\left[W(G_\theta^* (X_i, \hat{A}_i))-W(\hat{G}_{\theta}(X_i,\hat{A}_i))\right]
\le
C \frac{M}{k}\sqrt{\frac{v^\theta_{x,\hat a}}{n}}
\end{equation}

\textit{\textbf{Combining the two bounds.}}
Combining the upper bounds on $I$ and $II$ yields
\begin{equation}
\sup_{P\in\mathcal P}
\mathbb E_{P^n}\!\left[W(G^*_{\theta}(X_i,A_i))-W(\hat{G}_{\theta}(X_i,\hat{A}_i))\right]
\le
C\frac{M}{k}\sqrt{\frac{v^\theta_{x,\hat a}}{n}}
+
M\kappa L_s \operatorname{rMSE}_m(\hat{A}_i)
\end{equation}
which proves the claim.
\end{proof}

\begin{proof}[\textbf{Proof of Proposition \ref{prop:ha_lower_ext}}]
Let $X_i$ be such that there exists a function $m:\mathcal{X}\rightarrow[-\frac{1}{\kappa},\frac{1}{\kappa}]$ such that:
\begin{equation}
m(X_i) \sim \mathrm{Unif}\!\left[-\frac{1}{\kappa},\frac{1}{\kappa}\right]
\end{equation}
Define latent heterogeneity as
\begin{equation}
A_i := m(X_i) + U_i,
\qquad
U_i \sim \mathrm{Unif}[-r,r],
\qquad
U_i \perp m(X_i)
\end{equation}
where $r>0$ will be chosen below.
Define potential outcomes by
\begin{equation}
Y_i(0):=-\frac{M}{2}\operatorname{sign}(A_i),
\qquad
Y_i(1):=+\frac{M}{2}\operatorname{sign}(A_i)
\end{equation}
(with any convention $\operatorname{sign}(0)\in\{-1,+1\}$).
Then $|Y_i(d)|\le M/2$, so Assumption \ref{ass:kt18}.1 holds, and
\begin{equation}
\tau_i = Y_i(1)-Y_i(0)=M \operatorname{sign}(A_i)
\end{equation}
Let $D_i\sim\mathrm{Bernoulli}(p)$ independent of $(X_i,A_i,Y_i(0),Y_i(1))$
with $p\in(k,1-k)$, so Assumptions \ref{ass:kt18}.2--3 hold.

Suppose the proxy is constructed as $\hat A_i=\hat{f}_m(X_i)$ where $f$ is learned on an auxiliary sample of size $m$ independent of the policy-learning sample and then treated as fixed.
The population-optimal mapping (in mean squared error) satisfies
\begin{equation}
f^*(X_i)
\in
\arg\min_{f:\mathcal X\to\mathbb R}
\mathbb E_P[(A_i-f(X_i))^2]
\end{equation}
Since $\mathbb E[U_i]=0$ and $U_i\perp X_i$, the minimizer is
\begin{equation} \label{eq:optimal_f}
f^*(X_i)=\mathbb E[A_i\mid X_i]=m(X_i)
\end{equation}
Define $\hat A_i:=f^*(X_i)=m(X_i)$.
Then the proxy error equals
\begin{equation}
\varepsilon_i:=\hat A_i-A_i=-U_i
\end{equation}
so
\begin{equation}
\mathbb E[\varepsilon_i^2]=\mathbb E[U_i^2]=\frac{r^2}{3},
\qquad
\sqrt{\mathbb E[\varepsilon_i^2]}=\frac{r}{\sqrt{3}}
\end{equation}

Consider the threshold rule $g(a)=\mathbf 1\{a\ge 0\}$,
which belongs to $\mathcal G^\theta_{x,a}$ and satisfies
Assumption \ref{ass:score_funct} with margin constant $\kappa$
and Lipschitz constant $L_s=1$.

The oracle rule that observes $A_i$ treats iff $A_i\ge 0$ and attains
\begin{equation}
W(G^*_{\theta}(X_i,A_i))=\frac{M}{2}
\end{equation}
Given only $(X_i,\hat A_i)$, the Bayes-optimal feasible rule is
\begin{equation}
G_\theta^* (X_i, \hat{A}_i)
=
\mathbf 1\{\hat A_i\ge 0\}
=
\mathbf 1\{m(X_i)\ge 0\}
\end{equation}

We compute its misclassification probability.
Conditional on $m:=m(X_i)$,
\begin{equation}
\mathbb P(A_i<0\mid m)
=
\mathbb P(m+U_i<0\mid m)
=
\begin{cases}
\frac{r-m}{2r}, & 0\le m< r,\\
0, & m\ge r
\end{cases}
\end{equation}
and symmetrically for $m<0$.
Hence
\begin{equation}
\mathbb P\!\left(
G_\theta^* (X_i, \hat{A}_i)
\neq
\mathbf 1\{A_i\ge 0\}
\right)
=
\mathbb E\!\left[
\frac{r-|m(X_i)|}{2r}
\mathbf 1\{|m(X_i)|<r\}
\right]
\end{equation}

Since $m(X_i)\sim\mathrm{Unif}[-1/\kappa,1/\kappa]$,
its density on $[-1/\kappa,1/\kappa]$ equals $\kappa$.
For $r\le 1/\kappa$,
\begin{align}
\mathbb P\left( G_\theta^* (X_i, \hat{A}_i) \neq \mathbf 1\{A_i\ge 0\} \right)
&=
\int_{0}^{r}
\frac{r-u}{2r} \kappa du
=
\frac{\kappa}{2r}
\left[
ru-\frac{u^2}{2}
\right]_{0}^{r} \\
&=
\frac{\kappa r}{4}
\end{align}

Substituting $r=\sqrt{3}\sqrt{\mathbb E[\varepsilon_i^2]}$ gives
\begin{equation}
\mathbb P\left( G_\theta^* (X_i, \hat{A}_i) \neq \mathbf 1\{A_i\ge 0\} \right)
=
\frac{\kappa\sqrt{3}}{4}
\sqrt{\mathbb E[\varepsilon_i^2]}
\end{equation}

Therefore,
\begin{equation}
W(G(Z_i))
=
\frac{M}{2}
\left(
1-2 \mathbb P\!\left(
G(Z_i)\neq \mathbf 1\{A_i\ge 0\}
\right)
\right)
\end{equation}
so
\begin{equation}
W(G^*_{\theta}(X_i,A_i))-W(G_\theta^* (X_i, \hat{A}_i))
=
\frac{\kappa\sqrt{3}}{4}
M
\sqrt{\mathbb E[\varepsilon_i^2]}
\end{equation}

The remainder of the minimax lower-bound proof follows the same steps as in the measurement-error case:
(i) invoke the VC lower bound of \citet[Theorem 2.2]{kitagawa_who_2018} to obtain the statistical term
$C_K \frac{M}{k}\sqrt{v_{x,\hat a}/n}$ on a finite subclass in which $G^*_{\theta}(X_i,A_i)=G_\theta^* (X_i, \hat{A}_i)$,
and (ii) combine the proxy-information and statistical lower bounds to conclude
\begin{equation}
\inf_{\{\hat{G}_{\theta}(X_i,\hat{A}_i)\}}
\sup_{P\in\mathcal P(\rho)}
\mathbb E_{P^n}\!\left[W(G^*_{\theta}(X_i,A_i))-W(\hat{G}_{\theta}(X_i,\hat{A}_i))\right]
\ge 
C_K \frac{M}{k}\sqrt{\frac{v^\theta_{x,\hat a}}{n}} + C_4 M\kappa \rho
\end{equation}
where $C_4:= \sqrt{3}/8$.
\end{proof}

\end{document}